\documentclass[11pt]{article} 

\usepackage{amsmath,amssymb,amsthm,amsfonts,latexsym,bbm,xspace,graphicx,float,mathtools}
\usepackage[letterpaper,margin=1in]{geometry}
\usepackage{enumitem}

\usepackage{amsfonts}
\usepackage{amssymb,amsmath,amsthm}
\usepackage{epsfig}
\usepackage{color}
\usepackage{latexsym}
\usepackage{algorithm}
\usepackage{algorithmic}

\usepackage{graphicx} 

\title{A Simple and Approximately Optimal Mechanism for an Additive Buyer}
\date{} 

\author{Moshe Babaioff\footnote{Microsoft Research, moshe@microsoft.com}, Nicole Immorlica\footnote{Microsoft Research, nicimm@microsoft.com}, Brendan Lucier\footnote{Microsoft Research, brlucier@microsoft.com}, and S. Matthew Weinberg\footnote{Princeton University, smweinberg@princeton.edu. Work done in part while the author was supported by a Microsoft Graduate Research Fellowship and NSF CCF-6923736.}}

\newtheorem{theorem}{Theorem}

\newtheorem{corollary}{Corollary}
\newtheorem{lemma}{Lemma}

\newtheorem{proposition}{Proposition}

\newtheorem{definition}{Definition}

\newtheorem{observation}{Observation}

\newtheorem{openproblem}{Open Problem}
\newcommand{\poly}{\textsc{poly}}
\newenvironment{prevproof}[2]{\noindent {\em {Proof of {#1}~\ref{#2}:}}}{$\Box$\vskip \belowdisplayskip}

\newcommand{\cout}[1]{}

\newcommand{\mattnote}[1]{\textcolor{magenta}{#1}}

\newcommand \BRev{\text{\textsc{BRev}}}
\newcommand \SRev{\ensuremath{\textsc{SRev}}}
\newcommand \PRev{\ensuremath{\textsc{PRev}}}
\newcommand \Rev{\ensuremath{\textsc{Rev}}}
\newcommand \Welfare{\ensuremath{\textsc{Val}}}
\newcommand \Var{\ensuremath{\text{var}}}
\newcommand{\er}{\textsc{ER}}

\begin{document}
\maketitle

\begin{abstract}
We consider a monopolist seller with $n$ heterogeneous items, facing a single buyer.
The buyer has a value for each item drawn independently according to (non-identical) distributions, and her value for a set of items is additive. The seller aims to maximize his revenue. 


We suggest using the {\em a-priori better} of two simple pricing methods: \emph{selling the items separately}, each at its optimal price, and \emph{bundling together}, in which the entire set of items is sold as one bundle at its optimal price.
We show that for {any} distribution, this mechanism achieves a constant-factor approximation to the optimal revenue.
 Beyond its simplicity, this is the first computationally tractable mechanism to obtain a constant-factor
approximation for this multi-parameter problem. We additionally discuss extensions to multiple buyers and
to valuations that are correlated across items.


\end{abstract}

\section{Introduction}
A monopolist seller has a collection of $n$ different items to sell and is facing a single buyer.
How should he\footnote{Throughout this paper we shall refer to the seller as `he/him' and the buyer as `she/her.'} sell the items to maximize revenue given that the buyer's valuation is private and the buyer acts strategically? When there is only a single item, and a single buyer with private value drawn from a known distribution $F$, seminal work of Myerson \cite{Myerson81} shows that the optimal sale protocol is straightforward: simply post a take-it-or-leave-it price $p$, chosen to maximize the expected revenue, $p\cdot(1-F(p))$. This simple deterministic mechanism is optimal among all protocols, including interactive and randomized ones. Myerson's elegant theory  extends to single-item auctions with multiple buyers, and further generalizes to all ``single-dimensional'' domains.\footnote{Note that results for a single buyer also hold when there are many buyers but there are no supply constraints. In particular, they hold for digital goods (like software and music) that can be duplicated essentially for free.} 

Unfortunately, this rich theory fails to extend even to the simplest multi-item settings, where it was shown that optimal auctions might necessitate randomization~\cite{Thanassoulis04} and menus of infinite size~\cite{VincentM07,DaskalakisDT17}, exhibit non-monotonicity~\cite{HartR15}, or be computationally intractable~\cite{DaskalakisDT14} (see Section~\ref{sec:examples} immediately following the introduction for examples illustrating these issues). This is troubling not only from the perspective of analyzing optimal mechanisms, but also from the point of view of their usefulness.  For an auction to be useful in practice, it should be simple to describe and transparent in its execution. 
The danger, then, is that mechanisms that obtain exact revenue optimality in theory, regardless of complexity, might share the fate of other mathematically optimal designs and be rarely used in practice~\cite{Ausubel06}. It is therefore crucial to pair the study of revenue optimization with an exploration of the (approximation) power of simple auctions.  In other words, \emph{what is the relative strength of simple mechanisms versus complex optimal ones?} 

This question was first asked in seminal work of Hartline and Roughgarden~\cite{HartlineR09} for single-item auctions, and by seminal works of Chawla, Hartline, and Kleinberg~\cite{ChawlaHK07} and Hart and Nisan~\cite{HartN17} in multi-item settings. 
Despite making no direct reference to computation, this paradigm is not unlike those underpinning the theory of approximation algorithms: \emph{optimal} algorithms for NP-hard problems often do not yield much tractable insight, yet approximation algorithms for these same problems often do. The formal notion of ``polynomial-time'' generally (but not always) separates insightful algorithms from the rest. Similarly, \emph{optimal} auctions for multi-parameter domains do not yield much tractable insight, yet approximately-optimal auctions might (indeed, the thesis of this paper is that they do). Formal notions of complexity (discussed in Section~\ref{sec:examples}) again generally separate insightful auctions from the rest, but the less formal eyeball-test for simplicity generally does the trick as well.


\subsection{A Single Additive Buyer}
As discussed above, the goal of this work is to deepen our understanding of the problem of revenue optimization in multi-parameter settings through the lens of approximation. We follow Hart and Nisan~\cite{HartN17} in focusing on the simplest such multi-parameter problem: a seller owns $n$ different items and has no value for them (nor cost for production), and a single buyer has non-negative value $v_i$ for each item $i$, and that value is unknown to the seller. The seller only knows that each $v_i$ is sampled independently from a distribution $D_i$. The buyer is interested in maximizing her quasi-linear utility (value received minus price paid), while the seller sets the allocation and payment rules with the goal of maximizing expected revenue.\footnote{Note that the seller is constrained to always offer the buyer the option to ``stay home'' --- receive no items and pay nothing.} Both seller and buyer are risk neutral (the seller aims to maximize the expected payment, and the buyer aims to maximize her expected utility).    


Hart and Nisan consider the case of an \emph{additive} buyer, meaning that her value for a set $S$ of items is $\sum_{i\in S} v_i$.  As they point out, additivity is a natural starting point for this endeavor because, at first glance, the additive multi-item problem appears to be just a product of $n$ separate single-item problems. After all, the buyer's value for item $i$ does not depend at all on which other items she receives: it is always $v_i$, no matter what, by additivity. Moreover, any information the seller possesses about the buyer's values for items other than $i$ tells him absolutely nothing about her value for item $i$: it is always sampled from $D_i$, no matter what, by independence. Therefore, there is truly no interaction between the items whatsoever from the buyer's perspective, and it tempting to think  that 
the optimal mechanism should simply treat this as $n$ separate single-item problems. Myerson's theory tells us exactly how to solve each single-item problem: simply set a price $p_i \in \arg\max\{p \cdot (1-F_i(p))\}$ on each item, and let the buyer purchase whatever subset of items she likes. Somewhat counter-intuitively, this mechanism is actually \emph{not} optimal, even for extremely simple examples.

\paragraph{Example One: Two i.i.d.~Items~\cite{HartN17}.} Consider the case of $n=2$, where the buyer's value for each item is drawn i.i.d.~from the uniform distribution over the finite set $\{1,2\}$. Then treating each item as a separate problem, Myerson's theory would say that the seller should post a price of $1$ on each item, obtaining expected revenue one per item and two in total.\footnote{Note that setting a price of $2$ on each item guarantees the same revenue. The existence of multiple optimal prices doesn't drive the example, and the same phenomenon would happen with e.g. uniform distribution over the set $\{2, 5\}$ instead.} However, there is a better mechanism that offers the buyer only two choices: receive \emph{both} items together for a price of $3$, or nothing at all (for free). With probability $3/4$ the buyer chooses to purchase the bundle and pay 3, yielding expected revenue $9/4 > 2$. So despite the initial intuition, it is indeed suboptimal to solve the problem by stitching together solutions to the separate single-item problems.

To get further intuition behind this example, consider the case of large $n$, where the buyer's value for each item is still drawn i.i.d.\ from the uniform distribution over the set $\{1,2\}$. Then optimally selling each item separately still yields expected revenue of $n$. Consider instead offering again only two options: receive \emph{all} items together for a price of $(3/2-\varepsilon)n$, or nothing at all (for free), for an appropriately chosen $\varepsilon>0$. Now, except with probability exponentially small in $n\cdot \varepsilon^2$, the buyer is indeed willing to purchase the grand bundle, yielding expected revenue approaching $3n/2$ as $n$ grows large (for a multiplicative gap approaching $3/2$). Hart and Nisan~\cite{HartN17} show how to modify this example to exhibit a gap of $\Omega(\log(n))$ by replacing the uniform distribution with an Equal-Revenue distribution.\footnote{The Equal-Revenue distribution has CDF $F(x) = 0$ for $x \leq 1$, and $F(x) = 1-1/x$ for $x \geq 1$. The expected revenue obtained by posting any price $p\geq 1$ is one, thus any two such prices obtain equal revenue.}

What drives this example? Although there is no interaction between the items from the buyer's perspective, this is not the case from the seller's perspective. Indeed, the additional items enriches the strategy space of the seller and enables him to price options for which the buyer's value has much lower variance, allowing the seller to extract more of the buyer's welfare as revenue. At this point, we conclude not only that the multi-item problem is indeed richer than a product of single-item problems, but that treating it as such may come at significant cost: {selling the items separately does not guarantee any constant-factor approximation to the optimal attainable revenue}. 

In this simple example, the seller can gain by combining all the items together into a single grand bundle, to be sold at a take-it-or-leave-it price.
On the other hand, Hart and Nisan further show that this bundling approach can also be highly suboptimal.  Indeed, it guarantees only an $\Omega(n)$ approximation to the optimal revenue in general, considering the case where each value $v_i$ is independently $2^i$ with probability $2^{-i}$ (and $0$ otherwise). In this example, bundling together achieves revenue $O(1)$, while selling separately achieves revenue $n$. We must therefore further conclude that \textbf{neither selling separately nor bundling together guarantees 
a constant fraction of the optimal revenue.} The reader should notice, however, that in the above examples, the superior mechanism witnessing that selling separately does not achieve a constant-factor approximation is bundling together, and vice versa. 


\subsubsection{Further Complexity and the Need for Approximation}
So at this point, it is clear that optimal mechanisms are richer than one might naively expect. Still, one might reasonably hope that optimal mechanisms are not too complex. 
Unfortunately this is not the case, and prior work has identified numerous complexities that the optimal mechanism may possess. We briefly highlight one aspect below (menu complexity), and elaborate on the others (computational intractability, non-monotonicity) in Section~\ref{sec:examples}. 

\paragraph{Example Two: Randomization~\cite{DaskalakisDT14}.} Consider again the case of $n=2$, where the buyer's value for each item is drawn independently, but not identically.  Value $v_1$ is uniformly drawn from $\{1,2\}$, and $v_2$ is uniformly drawn from $\{1,3\}$. Quick calculations confirm that the revenue achieved by treating the items separately is $1+1.5 = 2.5$ (by setting prices of $1$ and $3$, respectively). The revenue achieved by bundling the items together is $3\cdot 3/4 = 2.25$ (by setting a price of $3$ which sells with probability $3/4$). But there is a better mechanism that offers the buyer three choices: (a) receive both items and pay $4$; (b) receive item one with probability one and item two with probability $1/2$ and pay $2.5$; or (c) receive nothing and pay nothing. Then when the buyer has valuations $\langle 1,3\rangle$ or $\langle 2,3 \rangle$, she will choose to pay $4$. When her values are $\langle 2,1 \rangle$ she will choose to pay $2.5$. When her values are $\langle 1, 1 \rangle$ she will choose to pay $0$, for a total expected revenue of $4\cdot 1/2 + 2.5\cdot 1/4 = 2.625 > 2.5$. This randomized mechanism  is in fact the \emph{unique} optimal mechanism\footnote{Note that the mechanism could be needlessly modified to additionally offer, for instance, an option to receive item one for $100$. Such an option would of course never be purchased. By ``unique,'' we mean that every optimal mechanism results in the buyer purchasing both items for $4$ when $v_2 =3$, item one w.p.~one and item two w.p.~$1/2$ when $\langle v_1,v_2\rangle = \langle 2,1\rangle$, and purchasing/paying nothing otherwise.} for this instance, and it beats any deterministic mechanism.

Example Two shows that even for simple instances, optimal mechanisms must sometimes be randomized. Yet, that optimal mechanism uses only one non-trivial lottery, 
so perhaps the degree of randomization required is not \emph{that} bad. Consider further the following example:

\paragraph{Example Three: Uncountable Menu Complexity~\cite{DaskalakisDT17}.} Consider again the case of $n=2$, but this time the buyer's value for each item is drawn i.i.d. from the distribution supported on $[0,1]$ with density $f(x) = 2(1-x)$.\footnote{Note that this is the Beta(1,2) distribution.} Daskalakis et al. prove for this example that the unique (up to differences of measure zero) optimal mechanism has \emph{uncountable menu complexity}~\cite{DaskalakisDT17}. That is, the number of distinct options available for the buyer to purchase is uncountable. They show that the optimal mechanism contains the following four kinds of options: (a) the buyer can receive item one with probability $1$, and item two with probability $\frac{2}{(4-5x)^2}$ paying price $\frac{2-3x}{4-5x} +\frac{2x}{(4-5x)^2}$, for \emph{any} $x \in [0,\approx .0618)$; (b) the buyer can receive item two with probability $1$, and item one with probability $\frac{2}{(4-5x)^2}$ paying price $\frac{2-3x}{4-5x} +\frac{2x}{(4-5x)^2}$, for any $x \in [0,\approx .0618)$; (c) the buyer can receive both items and pay $\approx .5535$; or (d) the buyer can receive neither item and pay nothing.

One can reasonably debate the finer points on the above object's mathematical tractability --- on one hand it can at least be succinctly described, on the other hand it offers an uncountable menu --- but what is made clear by this example is that any theory of optimal multi-item mechanisms must include a study of such auctions, which are significantly more complex than their single-item counterparts. At this point we must further conclude that \textbf{the study of truly simple multi-item auctions is unlikely to develop through a theory of \emph{exactly} optimal auctions.} 

\subsection{Our Results} 

\subsubsection{Main Result}
Let us now return to selling separately and bundling together, arguably the two simplest multi-item auctions,\footnote{They are both formally a ``black-box reduction'' to the single-item case.} neither of which guarantees a constant-fraction of the optimal revenue in all instances. Our main result is that the \emph{a priori maximum} of the revenue generated by these two approaches --- selling separately or bundling together --- does indeed guarantee a constant-factor approximation to the optimal revenue.  

\bigskip

\noindent
\textbf{Main Result} (Informal).
\emph{
In any market with a single additive buyer with independent item values, either selling separately or bundling together guarantees a 6-approximation to the optimal revenue.
}
\bigskip

Prior to our work, it was unknown whether \emph{any} simple mechanism achieves a constant-factor approximation to the optimal revenue, let alone one of these especially simple mechanisms. Moreover, note that our work further implies (Appendix~\ref{app:polytime}) that even if the seller does not know the distributions exactly but rather only has some reasonable access to them, a constant-factor approximation can be found in poly-time. Prior to our work, it was unknown whether a constant-factor approximation could be obtained by \emph{any} poly-time mechanism, even without the additional restriction of simplicity. Further, as selling separately and bundling together are both deterministic, our work further shows a constant multiplicative gap between the optimal deterministic and randomized mechanisms. Prior to our work, it was also unknown whether any deterministic mechanism could provide such a guarantee, even without the additional restriction of simplicity. Additionally, both mechanisms have desirable robustness properties (overviewed further in Section~\ref{sec:related}). 



\paragraph{Brief Intuition.} To get a sense of why our main result holds, recall the discussion following Example One with $n$ i.i.d.\ distributions.  We noted that if the sum of the buyer's values for all the items tends to concentrate around its expectation, then bundling together (at the right price --- slightly below the expected value) will extract a significant fraction of this total value as revenue.  The bundling-together mechanism will therefore be approximately optimal when the total value concentrates.  On the other hand, since the total value is a sum of independent random variables, 
it \emph{will} concentrate unless the sum is dominated by rare events where one or more items have significantly higher-than-expected value.  How should we handle such cases?  Intuitively, if these ``tail'' events are indeed rare, then we are unlikely to see many of them at once and the optimal revenue will be driven by a small number of individual items.  Since selling items separately optimizes the revenue from each individual item by itself, the sell-separately mechanism is a good candidate for covering this tail case.

To make this intuition concrete, a surprisingly challenging question is to properly define ``concentration'' and ``tail events.'' For example, for $n$ i.i.d. Equal Revenue distributions, the expected value for each item is infinite, yet the expected optimal revenue (of selling all items optimally) is finite. In particular, as bundling together (but \emph{not} selling separately) is approximately optimal for $n$ i.i.d. Equal Revenue distributions, we would like our argument to claim that the sum of $n$ i.i.d. Equal Revenue distributions ``concentrates,'' despite the fact that its expected value is infinite. Challenges similarly arise when trying to properly define ``tail events.''

This led Li and Yao to develop a Core-Tail Decomposition~\cite{LiY13}, which proposes adequate definitions for tail events, and formally separates the analysis into cases where a tail event occurs (``the tail'') and the rest (``the core''). One key difference between our approach and that of Li and Yao is the definition of tail events, and we briefly highlight a top-level distinction here.  Consider starting from an arbitrary instance with $n$ items, and adding to that instance $n'$ items whose values are drawn from a point-mass at $0$. Clearly, this does not change the underlying instance. Yet, this modification \emph{does} change the definitions for tail events proposed in~\cite{LiY13}, thereby changing the analysis. Our choice of decomposition has the property that the analysis is invariant under this modification, as discussed in Section~\ref{sec:one-ind}, and our approach enables a tighter analysis (and thereby improving their approximation guarantees as well --- see Section~\ref{sec:related}). All further distinctions, including formal definitions, are deferred to the technical sections.

\subsubsection{Additional results}
Beyond the core single-buyer setting, we consider several extensions using similar techniques. First we consider the \emph{multi-buyer setting}, where each buyer's value for each item is drawn independently (not necessarily identically) from arbitrary distributions. Here, we show again that selling separately (that is, sell each item using Myerson's optimal single-item auction) achieves an $O(\log n)$-approximation. Note that this is asymptotically tight, even for the case of a single buyer (via the modification of Example One in~\cite{HartN17}). Prior to our work, no non-trivial bounds were known on the approximation guarantee of any class of mechanisms for this setting. 

We further show that in contrast to the single-buyer case, the better of selling separately and bundling together (treating the grand bundle of all items as a single item and running Myerson's optimal auction) does \emph{not} guarantee an $o(\log n)$ approximation in the multi-buyer setting.  We further extend this lower bound to the class of \emph{partition mechanisms}, which partition the items into disjoint subsets and run Myerson's optimal auction for each subset separately (generalizing both selling separately and bundling together). A preliminary version of this paper presented at FOCS 2014 left as its main open problem whether a simple (and/or deterministic, poly-time) mechanism could guarantee a constant-factor approximation in the multi-buyer case, which was resolved by Yao~\cite{Yao15} (discussed further in Section~\ref{sec:related}). 

Finally, in order to better understand this new class of partition mechanisms, we study the performance of selling separately and bundling together against the optimal partition mechanism in a variety of settings. For multiple buyers with independent items, we show that the better of selling separately and bundling together achieves a constant-factor approximation to the optimal partition mechanism when \emph{either} buyers (Theorem~\ref{thm:previid}) \emph{or} items (Theorem~\ref{thm:previiditems}) are i.i.d., but no better than an $\Omega(\log n)$-approximation in general (Proposition~\ref{prop:lb-prev-max}). We consider also the single buyer case when values for the items may be \emph{arbitrarily correlated}. While neither class of mechanisms can guarantee any non-zero fraction of the optimal revenue for even a single buyer (\cite{BriestCKW15,HartN13}), the question remains as to whether simple mechanisms can approximate more complex (though still suboptimal) mechanisms in the presence of correlation. To this end, we prove that selling items separately obtains an $O(\log n)$-approximation to the optimal obtainable revenue by a partition mechanism, and that this is tight. In fact, we show a gap of $\Omega(\log n)$ between the \emph{a priori} better of selling separately and bundling together, versus the optimal partition mechanism. We include several tables in Appendix~\ref{app:tables} displaying the relative power of the various classes of mechanisms studied in this paper, noting here that as of our work, all upper and lower bounds are (asymptotically) matching.

\section{Complexity of Optimal Multi-Item Mechanisms}\label{sec:examples}
Here, we'll describe (without proofs) several known examples from the literature that further demonstrate the complexity of optimal multi-item auctions, even with just a single buyer whose value for the items is additive. In each, we'll use $D_i$ to denote a distribution, and $F_i$ its CDF.

\paragraph{Example Four: Revenue Non-Monotonicity~\cite{HartR15}.} Consider two distributions $D = \times_{i\in [n]} D_i$, and $D^+ = \times_{i\in [n]} D_i^+$, where each $D_i^+$ \emph{stochastically dominates} $D_i$ (that is, $F_i(x) \geq F^+_i(x)$ for all $i, x$). This implies that draws $(v, v^+)$ from $(D, D^+)$ can be coupled so that $v(S) \leq v^+(S)$ for all $S$ with probability $1$. It seems natural to conjecture that the optimal achievable revenue of $D^+$ should exceed that of $D$, as this is true for $n=1$ (the single item case) and the proof is a trivial corollary of~\cite{Myerson81}.\footnote{Specifically, for any price $p$ that might be set, $p \cdot (1-F^+(p)) \geq p \cdot (1-F(p))$.} However, Hart and Reny provide explicit distributions $D, D^+$ \emph{with i.i.d. marginals for two items} such that the optimal achievable revenue for $D$  \emph{strictly exceeds} that for $D^+$~\cite{HartR15}. That is, the optimal revenue for a ``strictly better'' distribution is strictly worse. This property has been termed \emph{revenue non-monotonicity}, and provides further evidence of the complexity of optimal multi-item auctions. Indeed, one could imagine $D^+$ resulting from $D$ after an advertising campaign which increases the values of all consumers in a population for all items. Revenue non-monotonicity implies that such a campaign may not only harm the revenue of an existing auction, but also harm the optimal achievable revenue (and even in cases where $D, D^+$ are both i.i.d. over two items).

\paragraph{Example Five: Computational Intractability~\cite{DaskalakisDT14}.} Consider being given as input a discrete product distribution over valuations for $n$ items, where each value $v_i$ is drawn from $D_i$ and the support of each $D_i$ is of size \emph{two}. That is, each $v_i$ is either $a_i$ or $b_i$, and is $a_i$ with probability $p_i$, all of which are rational and of bit complexity $\poly(n)$. Then the entire input can be described by these $3n$ numbers, and has size $\poly(n)$. It is somewhat tricky to formalize exactly what it should mean to ``find'' the optimal mechanism (since it may, for instance, have exponential menu complexity), but Daskalakis et al.~\cite{DaskalakisDT14} prove that it is \#P-hard to find the optimal auction in this setting in the following strong sense: unless $\text{ZPP} \supseteq \text{P}^{\text{\#P}}$, no randomized poly-time procedure can take as input $(\vec{a},\vec{b},\vec{p})$, and a further input valuation vector $\vec{v}$, and guess (correctly with probability at least $1/2+1/\poly(n)$) whether the optimal mechanism awards the buyer with valuation $\vec{v}$ item one with probability $0$ or $1$ (if neither is true, the algorithm is allowed to behave arbitrarily). This rules out any reasonable poly-time solution, as any notion of a ``solution'' must be able to determine which items are purchased by a given realized $\vec{v}$. 

\paragraph{Example Six: Infinite Gaps with Correlation~\cite{BriestCKW15, HartN13}.} Finally, consider the case that the buyer's values for the items are \emph{correlated}. That is, there is an arbitrary $n$-dimensional distribution $D$ over $\mathbb{R}_+^n$, and the buyer's values for all $n$ items are drawn jointly from $D$. In this case, Hart and Nisan~\cite{HartN13} provide an explicit distribution $D$ over $n=2$ items such that the revenue of the optimal mechanism is \emph{infinite}, yet the revenue of any mechanism of menu complexity $C$ is at most $C$. It was further observed that this same $D$ is stochastically dominated by some $D^+$ (that is, couples $(\vec{v},\vec{v}^+)$ can be drawn from $(D, D^+)$ so that $v(S) \leq v^+(S)$ for all $S$ with probability $1$) such that the optimal revenue for $D$ is infinite, yet the optimal revenue for $D^+$ is at most one~\cite{RubinsteinW15}.\footnote{In fact, their $D^+$ simply draws $(v_1,v_2)\leftarrow D$ and sets $(v_1^+,v_2^+) = (\max\{v_1,v_2\},\max \{v_1,v_2\})$.}

Further related work is described in Section~\ref{sec:related}, but we have surveyed  
these examples 
to emphasize the following themes:
\begin{enumerate}
\item Examples One, Two, and Five highlight that the optimal mechanism may be surprisingly complex, even in extremely simple examples. This suggests that we are unlikely to make progress by simply restricting the allowable input distributions.
\item Examples Three, Four, and Five highlight the three main complexities commonly associated with optimal auctions that are viewed as impractical: unbounded menu complexity, non-monotonicity, and computational intractability. Note that in contrast, the approximately-optimal \emph{a priori} maximum of selling separately and bundling together is deterministic, has polynomial menu complexity,\footnote{Technically, selling separately has exponential menu complexity, but~\cite{BabaioffGN17} prove that one can get arbitrarily close to the revenue of selling separately with polynomial menu complexity.} is revenue-monotone,\footnote{To conclude revenue-monotonicity, observe that both selling separately and bundling together are products of single-item auctions, and therefore revenue-monotone. The maximum of revenue-monotone mechanisms is also revenue-monotone.} and implementable in poly-time.
\item In addition to the discussion preceding Example One, Example Six further motivates restricting attention to independent items, as even approximately optimal mechanisms can be arbitrarily complex without this.
\end{enumerate}

\section{Related Work}\label{sec:related}
\subsection{Prior Work}
Seminal work of Hartline and Roughgarden initiated the agenda of ``simple versus optimal'' mechanisms: the study of simple mechanisms through the lens of approximation~\cite{HartlineR09}. Their work considers single-dimensional settings, and shows that one can often approximate the revenue of Myerson's optimal auction with something even simpler. This agenda has even more bite in multi-dimensional settings, where optimal auctions are \emph{far} more complex (c.f. the above examples). On this front, seminal work of Chawla, Hartline, and Kleinberg considers a single unit-demand buyer,\footnote{A valuation function $v(\cdot)$ is unit-demand if $v(S) = \max_{i \in S}\{v(\{i\})\}$ for all $S$.} and prove that a deterministic item-pricing guarantees a constant-factor approximation to the optimal mechanism~\cite{ChawlaHK07}. Further follow-up work considers multiple buyers, but faces barriers in moving beyond unit-demand preferences~\cite{ChawlaHMS10, ChawlaMS10,ChawlaMS15, KleinbergW12}. 

Hart and Nisan first proposed studying a single additive buyer through the lens of approximation, as this is the simplest possible setting where all previously developed tools remained stuck~\cite{HartN17}. Their work provides several simple lemmas, whose composition proves surprisingly strong conclusions: they show that selling separately achieves an $O(\log^2 n)$-approximation, and further that when the items are i.i.d., bundling together achieves an $O(\log n)$-approximation. Follow-up work of Li and Yao introduces the \emph{Core-Tail Decomposition} technique, improving the guarantee of selling separately to $O(\log n)$ (which is tight) and the guarantee of bundling together when all items are i.i.d. to $O(1)$~\cite{LiY13}. Our work makes use of tools developed in both works, and improves the approximation guarantee to $6$ without any assumptions.

A related sequence of papers~\cite{CaiD11, CaiH13} use extreme value theorems to prove that when distributions satisfy the Monotone Hazard Rate (MHR) condition, nearly-optimal mechanisms can be found in poly-time and are fairly simple. 

\subsection{Subsequent Work}
An initial presentation of this work at FOCS 2014 posed three open problems, all of which have since been resolved (and then some). The first asked whether a simple mechanism could guarantee a constant-factor approximation for multiple additive buyers (as our lower bounds prove that partition mechanisms cannot achieve this guarantee), which was resolved by Yao~\cite{Yao15}. The main conceptual discovery is ``the right'' extension of bundling together for multiple buyers, which turns out to be an \emph{entry fee}. In the context of our work, bundling together can be interpreted as the mechanism which gives the buyer all items for free, as long as the buyer pays an entry fee to participate.\footnote{That is, there is an entry fee $p$. The buyer, with full knowledge of her valuation, decides whether or not to pay $p$. If she pays, she participates in the mechanism, which gives her all items for free. If not, she leaves with no items and pays nothing.}%
~\cite{Yao15} proves that the proper extension of our results to multiple buyers is to replace the ``give the buyer all items for free'' mechanism with the VCG mechanism (which sells each item separately using a second-price auction), but maintain the entry fee.\footnote{Now, the entry fee $p_j$ for buyer $j$  depends on the valuations of others (yet is still independent of the valuation of buyer $j$).} This idea persists in follow-up works (sometimes further replacing the VCG mechanism with other simple mechanisms)~\cite{ChawlaM16, CaiDW16, CaiZ17}. 


The second asked whether a simple mechanism could guarantee a constant-factor approximation for a single buyer who was neither unit-demand nor additive (e.g. $k$-demand, satisfying $v(S) = \max_{T \subseteq S, |T| \leq k} \{\sum_{i \in T} v(\{i\})\}$ for all $S$). This was resolved in~\cite{RubinsteinW15}, who show again that either selling separately or bundling together achieves a constant-factor approximation even when the buyer has \emph{subadditive} valuations over independent items.\footnote{Roughly speaking, ``independent items'' means that the random variables $v(S_1),\ldots, v(S_\ell)$ are independent whenever $S_i \cap S_j = \emptyset$ for all $i, j$. The formal definition is slightly (but strictly) more restrictive than this.} Further follow-up work extends this to models of limited complementarity~\cite{EdenFFTW17a}. 

The third asked whether our results could be extended to models of limited correlation, such as those considered in~\cite{ChawlaMS10,ChawlaMS15}. This was resolved in~\cite{BateniDHS15}, who show essentially that our approach is robust to linear combinations: if there are $k$ independent \emph{features} that a buyer might value, and the value of each item is a fixed linear combination of few features, then again selling separately or bundling together achieves a constant-factor approximation.

However, work has continued far beyond these specific open problems. Most notably, work of Chawla and Miller~\cite{ChawlaM16} and Cai and Zhao~\cite{CaiZ17} consider multiple buyers, all of whom are neither unit-demand nor additive, and prove that a posted-price mechanism with entry fee achieves a constant-factor approximation (i.e., posting prices is ``the right'' extension of selling separately, and adding an entry fee is ``the right'' extension of bundling together).~\cite{CaiZ17} is the state-of-the-art in this direction, which shows that these mechanisms guarantee a constant-factor approximation when buyers are fractionally subadditive over independent items, and an $O(\log n$)-approximation when buyers are subadditive over independent items. Cai et al.~\cite{CaiDW16} further show how to interpret our work and those following~\cite{ChawlaHK07} via the same dual solution in a duality framework. This unified presentation of both lines of work is also given in~\cite{hartlinebook}.

There are also numerous follow-up works in tangential directions, establishing that the mechanisms studied in these lines of work are quite robust: works of~\cite{MorgensternR16, BalcanSV16, BalcanSV18, CaiD17, Syrgkanis17, AzarKW19} prove that these mechanisms can be learned using only polynomially many samples from the underlying distributions, and~\cite{GoldnerK16} further shows that they can be made \emph{prior-independent} at the cost of additional constant factors. In a similar vein,~\cite{ChenLL17} proves that these results are fairly robust to the Bayesian assumption, and their guarantees hold when the buyers (and not the auctioneer) know the prior.~\cite{ChengGMW18} extend our analysis to accommodate a budget-constrained buyer.~\cite{Rubinstein16} provides a PTAS for the optimal partition mechanism (and proves that no FPTAS exists unless P = NP). Works of~\cite{Carroll17, GravinL18} target max-min guarantees (i.e. the auction which maximizes the minimum achieved revenue over all $D$ in some class $\mathcal{D}$) instead of worst-case approximation guarantees. Our tools have also found use in seemingly unrelated follow-ups studying gains from trade~\cite{BrustleCWZ17, BabaioffCGZ18}, information revelation~\cite{DaskalakisPT16, FuLLT18, ElmachtoubH17, CaiEFLWZ18}, and ``Bulow-Klemperer''-type~\cite{BulowK96} results~\cite{EdenFFTW17b, LiuP18,FeldmanFR18, BeyhaghiW19}. 

\subsection{Recent Work and Open Directions}
\paragraph{Nearly-simple and Nearly-optimal mechanisms.} Our work shows that a constant-factor approximation\footnote{We have originally proved that this constant is at most 7.5, and soon after, Rubinstein has tightened the analysis yielding an upper bound of 6. 
	The current best analysis show that this constant is at most $5.2$~\cite{MaS15}, yet at least $2$~\cite{Rubinstein16}.} can be found in poly-time, and implies that a constant-factor approximation can be achieved with polynomial menu complexity (the latter claim further requires a result of~\cite{BabaioffGN17}). Works discussed earlier prove that the optimum cannot be found in poly-time~\cite{DaskalakisDT14}, and may have uncountable menu complexity~\cite{DaskalakisDT17}, but do not rule out even an FPTAS. Works of~\cite{ChenDPSY14, ChenDOPSY15, ChenMPY18} identify similar complexity barriers for related problems. The major open question here is determining whether or not an FPTAS(/PTAS/QPTAS) exists. Extremely recent work now provides a QPTAS for a single unit-demand buyer~\cite{KothariMSSW19}, but the entire spectrum still remains open for additive buyers (along with providing/ruling out a FPTAS/PTAS for unit-demand). The equally significant question for menu complexity also remains open --- on this front, only recently was it shown that \emph{some bounded} (as a function of $n, \varepsilon$) menu complexity suffices to guarantee a $(1-\varepsilon)$-approximation~\cite{BabaioffGN17}, and the first non-trivial lower bounds were proved for the case of $n=2$~\cite{Gonczarowski18}. 

\paragraph{All the way to Subadditive.} Recent work discussed above proves that a posted-price mechanism with entry-fee guarantees a constant-factor approximation when buyers are \emph{fractionally subadditive with independent items}~\cite{CaiZ17}, and further shows that posted-price mechanisms alone guarantee an $O(\log n)$-approximation when buyers are subadditive. The key barrier to obtaining a constant-factor approximation for all subadditive valuations is the following: one portion of the analysis in~\cite{CaiZ17} makes use of an analysis in~\cite{FeldmanGL15} of posted-price mechanisms for \emph{welfare} guarantees. The analysis of~\cite{FeldmanGL15} guarantees an $O(1)$-approximation to the optimal welfare for fractionally subadditive buyers, and an $O(\log n)$-approximation when buyers are subadditive. It remains an open question whether the~\cite{FeldmanGL15} analysis can be improved to $O(1)$ for subadditive buyers (explicitly posed in~\cite{FeldmanGL15}), which would very likely extend~\cite{CaiZ17} to subadditive buyers as well.\footnote{Note that~\cite{CaiZ17} does not treat~\cite{FeldmanGL15} as a black-box, but analysis ``in the same spirit'' as~\cite{FeldmanGL15} would suffice.} Independently, it is generally an important open problem to extend~\cite{CaiZ17} all the way to subadditive valuations (explicitly posed in~\cite{CaiZ17} --- see therein for a deeper discussion).

\section{Preliminaries}\label{sec:prelim}
The setting we consider is that of a single monopolist seller with $n$ heterogeneous and indivisible items for sale to $m$ additive, risk-neutral, quasi-linear buyers.  That is, each buyer $j$ has a non-negative value $v_{ij}$ for item $i$.  While our main results are for the setting of a single buyer, we will define our setting more generally; this will be useful when discussing extensions.  If a randomized outcome awards buyer $j$ item $i$ with probability $\pi_{ij}$ and charges her a price $p_j$ in expectation, then her utility for this outcome is $\sum_i v_{ij} \pi_{ij} - p_j$. Each value $v_{ij}$ is sampled independently from a known distribution $D_{ij}$, supported on $\mathbb{R}_{\geq 0}$. We make no assumptions on $D_{ij}$ whatsoever. We refer to $D$ as the joint $(m\cdot n)$-dimensional distribution over all buyers' values for all items, $D_i$ as the $m$-dimensional distribution over all buyers' values for item $i$, and $D^j$ as the $n$-dimensional distribution over buyer $j$'s values for all items. Furthermore, we denote by $\vec{v}$ a random sample from $D$, $\vec{v}_i$ a random sample from $D_i$, and $\vec{v}^j$ a random sample from $D^j$.
We also denote the maximum value for item $i$ as $v^*_i = \max_j\{v_{ij}\}$.
We use the notation $D^{-j}$ to denote the distribution for all buyers but $j$ (and use similar notation for other vectors).

The revelation principle~\cite{Myerson81} establishes that the optimal revenue of \emph{any} (not necessarily truthful) auction at \emph{any} Bayes-Nash equilibrium is captured by a direct mechanism that is Bayesian Incentive Compatible (and we will therefore restrict attention to optimal mechanisms of this form, as it is w.l.o.g.). That is, this mechanism simply asks each buyer to report a value for each item, and it is in each buyer's interest to report their true value, assuming that all other buyers do so as well. All of the mechanisms we describe will also satisfy the stronger property of Dominant Strategy Incentive Compatibility, where it is in each buyer's interest to report their true values no matter the other buyers' behavior (observe further that these definitions coincide when there is just $m=1$ buyer). As usual, we also impose the individual rationality constraint, saying that every buyer's utility is non-negative when truthful. Formally:
\begin{itemize}
\item For a given direct mechanism, let $x_{ij}(\vec{v})$ denote the probability that item $i$ is awarded to buyer $j$ on bids $\vec{v}$, and $p_j(\vec{v})$ denote the price that buyer $j$ pays. Define also the interim variables $\pi_{ij}(\vec{v}^j) = \mathbb{E}_{\vec{v}^{-j}\leftarrow D^{-j}}[x_{ij}(\vec{v}^{-j};\vec{v}^j)]$, and $q_j(\vec{v}^j) = \mathbb{E}_{\vec{v}^{-j}\leftarrow D^{-j}}[p_j(\vec{v}^{-j};\vec{v}^j)]$. When there is just $m=1$ buyer, observe that $x = \pi$. 
\item A direct mechanism is \emph{Bayesian Incentive Compatible} (BIC) if for all $j, \vec{v}^j,\vec{w}^j$, it holds that  $\sum_i v_{ij} \cdot \pi_{ij}(\vec{v}^j) - q_j(\vec{v}^j) \geq \sum_i v_{ij}\cdot \pi_{ij}(\vec{w}^j) - q_j(\vec{w}^j)$.
\item A direct mechanism is \emph{Dominant Strategy Incentive Compatible} (DSIC) if for all $j, \vec{v}^j, \vec{w}^j, \vec{v}^{-j}$,  it holds that  $\sum_i v_{ij} \cdot x_{ij} (\vec{v}^{-j};\vec{v}^j) - p_j(\vec{v}^{-j};\vec{v}^j) \geq \sum_i v_{ij} \cdot x_{ij}(\vec{v}^{-j};\vec{w}^j) - p_j(\vec{v}^{-j};\vec{w}^j)$. When there is just $m=1$ buyer, observe that this is equivalent to BIC.
\item A direct mechanism is \emph{Ex-Post Individually Rational} (IR) if for all $j, \vec{v}^j, \vec{v}^{-j}$,  it holds that \\ $\sum_i v_{ij} \cdot x_{ij} (\vec{v}^{-j};\vec{v}^j) - p_j(\vec{v}^{-j};\vec{v}^j) \geq 0$. 
\end{itemize}



We further use the following terminology to discuss the revenue obtainable by various types of mechanisms, where the first three are taken from~\cite{HartN17}. Below, we reference Myerson's optimal single-item auction, which for $m=1$ buyer simply sets the price $p^* = \arg\max_{p} \{p \cdot \Pr[v \geq p]\}$~\cite{Myerson81}. The precise format of Myerson's auction for $m>1$ buyers is immaterial for our results, as we use no properties of this auction in our proofs other than its optimality.


\begin{itemize}
\item $\Rev(D)$: The optimal revenue (more precisely, the supremum of revenues) obtained by any (possibly randomized) BIC/IR mechanism when the buyer profile is drawn from $D$. 
\item $\SRev(D)$: The optimal revenue (more precisely, the supremum of revenues) obtained by selling items separately when the buyer profile is drawn from $D$. That is, the revenue obtained by running Myerson's optimal auction separately for each item. Recall that when there is just $m=1$ buyer, $\SRev(D)$ is achieved by setting a price $p_i$ on each item $i$, and letting the buyer pick any subset of items to purchase.
\item $\BRev(D)$: The optimal revenue (more precisely, the supremum of revenues) obtained by selling the grand bundle when the buyer profile is drawn from $D$. That is, the revenue obtained by running Myerson's optimal auction when treating the grand bundle as a single item. Recall that when there is just $m=1$ buyer, $\BRev(D)$ is achieved by setting a price $p$ on the grand bundle, and letting the buyer purchase the grand bundle for price $p$. 
\item $\PRev(D)$: The optimal revenue (more precisely, the supremum of revenues) obtained by any \emph{partition} mechanism when the buyer profile is drawn from $D$. That is, the maximal revenue obtained by first partitioning the items into disjoint bundles, and then running Myerson's optimal auction separately for each bundle, treating each bundle as a single item.
\end{itemize}

Observe that some of the terms above may not be well-defined if $\Rev(D)$ is unbounded. In this case, our proofs establish that $\SRev(D), \BRev(D), \PRev(D)$ are all unbounded as well (although we will not explicitly state these conclusions).  Also, for ease of exposition, we will only consider the case when the supremum of revenues is actually achieved by some mechanism, which we will refer to as the ``optimal mechanism.''  While we do not explicitly discuss cases where the supremum is not realized but is instead the limit of a sequence of mechanisms, we note that all of our proofs carry over by standard limiting arguments.

Observe also that selling separately and bundling together are both partition mechanisms, and that all partition mechanisms are in fact DSIC. Given a distribution $D$ over profiles, we will often consider the welfare $\sum_i v^*_i = \sum_i \max_j\{v_{ij}\}$ of a buyer profile $\vec{v}$ drawn from $D$. We will write $\Welfare(D)$ for the expected optimal welfare, so that $\Welfare(D) = \mathbb{E}_{\vec{v} \sim D}\left[\sum_i v^*_i\right]$.  We will also write $\Var(D) = \Var_{\vec{v} \sim D}(\sum_i v^*_i)$ for the variance of the welfare. Observe that, immediately from Individual Rationality, $\Rev(D) \leq \Welfare(D)$ for all $D$.

We will make use of some results from~\cite{HartN17} that provide useful bounds on $\Rev(D)$. We include proofs in Appendix~\ref{app:prior} for completeness. Lemma~\ref{lem:HN1} is stated and proved directly in~\cite{HartN17} and also in~\cite{CaiH13}. Lemma~\ref{thm:HN} is not directly stated nor proved, but is similar to an implicit result from~\cite{HartN17}.

In Lemma~\ref{lem:HN1} below, we think of $D$ and $D'$ as being distributions over values for disjoint sets of items, for the same set of $m$ buyers.  The distribution $D \times D'$ then draws values for those two sets of items, independently, from $D$ and $D'$ respectively.

\begin{lemma}~\label{lem:HN1}(\cite{HartN17,CaiH13}) $\Rev(D \times D') \leq \Welfare(D) + \Rev(D')$.
\end{lemma}

The next result establishes a weak bound on $\Rev(D)$ with respect to $\SRev(D)$.

\begin{lemma}\label{thm:HN} $\Rev(D) \leq n\cdot m\cdot \SRev(D)$.
\end{lemma}

%
%

\section{The Core-Tail Decomposition}\label{sec:core}


We make use of an idea developed by Li and Yao~\cite{LiY13} called the Core-Tail Decomposition of a value distribution for a single buyer. In order to obtain our stronger results for a single buyer and also extend to many buyers, we define the core differently but in the same spirit. The idea is to separate each $m$-dimensional value distribution for each item into the core and the tail, the tail being the part where some buyer has an unusually high value for the item. Then the core of the entire $nm$-dimensional distribution is the product of all the cores, and the tail is everything else.

\subsection{Defining the Core and Prior Results}

Below we formalize the notion of the core.  We introduce some notation that will be used throughout the paper. Many definitions below define distributions conditioned on events. Sometimes, these events will have probability $0$ of occurring. For simplicity of notation, we'll define a ``null'' distribution which deterministically outputs $0$, and replace any distribution conditioned on a zero probability event with the null distribution.


\begin{itemize}
\item \textbf{$r_i$:} The optimal revenue obtainable by selling just item $i$, using Myerson's optimal auction.
\item \textbf{$r$:} $\sum_i r_i$. The total revenue from optimally selling the items separately.  Note that $r=\SRev(D)$, but we introduce this redundant notation for convenience.
\item \textbf{$t_i$:} A parameter for item $i$, used to define the separation between the core and tail of distribution $D_i$.  We will think of $t_i$ as a multiplier applied to $r_i$.  The core for item $i$ will be supported on the interval $[0,t_i r_i]$, and the tail for item $i$ will be supported on $(t_i r_i, \infty)$.  Different results throughout the paper will specify different choices for $t_i$. We will often abuse notation and say that ``item $i$ is in the tail'' when $v^*_i > t_i r_i$ and ``item $i$ is in the core'' when $v^*_i \leq t_i r_i$.
\item \textbf{$p_i$:} $Pr[v^*_i > t_i r_i]$, the probability that the highest value on item $i$ lies in the tail. Note that this may be $0$, and also that this depends both on the distribution $D_i$, as well as the choice of $t_i$.
\item \textbf{$D_i^C$:} The core of $D_i$, the conditional distribution of $\vec{v}_i$ conditioned on $v^*_i \leq t_i r_i$. Note that this may be the null distribution if $p_i = 1$.
\item \textbf{$D_i^T$:} The tail of $D_i$, the conditional distribution of $\vec{v}_i$ conditioned on $v^*_i > t_i r_i$. Note that this may be the null distribution if $p_i = 0$.
\item \textbf{$A$:} Throughout our notation, we will use $A$ to represent a subset of items.  We often think of $A$ as the items whose values lie in the tail of their respective distributions.
\item \textbf{$D_A^T$:} $A$ is a subset of items, and $D_A^T$ is a product distribution equal to $\times_{i \in A} D_i^T$.
\item \textbf{$D_A^C$:} $A$ is a subset of items, and $D_A^C$ is a product distribution equal to $\times_{i \notin A} D_i^C$.  Notice that the product is over items \emph{not} in $A$.  We think of $D_A^C$ as representing the distribution of values in the core, conditional on $A$ being the set of items whose values lie in the tail.
\item \textbf{$D_A$:} $D_A^C \times D_A^T$. Note that this product is taken over the tail of items in $A$ and the core of items not in $A$.  In other words, $D_A$ is the distribution $D$, conditioned on $v_i^* > t_i r_i$ for all $i \in A$ and conditioned on $v_i^* \leq  t_i r_i$ for all $i \notin A$.
\item \textbf{$p_A$:} $(\prod_{i \in A} p_i)(\prod_{i \notin A} (1-p_i))$. When $D_A$ is not null, this equals $Pr[\vec{v} \in \text{support}(D_A)]$. 
\end{itemize}

Before stating our core-tail decomposition lemma, we present some known results about the core. The lemmas below were either stated explicitly in~\cite{LiY13} or~\cite{HartN17} for a single buyer, or use ideas from one of those papers. We put a citation in the statement of such lemmas, but include all proofs in Appendix~\ref{app:core} for completeness.




\begin{lemma}\label{lem:LY1}(\cite{LiY13}) $p_i \leq 1/t_i $ for all $i$.
\end{lemma}

\begin{lemma}\label{lem:LY2}(\cite{LiY13}) $\Rev(D_i^C) \leq r_i$ and $\Rev(D_i^T) \leq r_i/p_i$.
\end{lemma}

\begin{lemma}\label{lem:HN2}(\cite{HartN17}) $\Rev(D) \leq \sum_A p_A \Rev(D_A)$.
\end{lemma}

\subsection{The Core-Tail Decomposition Lemma}

In this section we state our Core-Tail Decomposition Lemma, which relates the optimal revenue from a distribution $D$ to the revenue and welfare that can be extracted from the tail and core of $D$, {respectively}. This result is similar in spirit to the main lemma of~\cite{LiY13}.  

Our first result, Lemma \ref{lem:easy}, is our main decomposition lemma.  The lemma states that the optimal revenue from distribution $D$ can be split into a contribution from the core of $D$ and a contribution from the tail of $D$.  One might hope for a bound of the form ``the optimal revenue from $D$ is at most the optimal revenue from the tail plus the optimal revenue from the core.'' Indeed, such a bound is attainable for a single buyer~\cite{LiY13}, but is problematic for many buyers. We will therefore settle for a weaker bound: the optimal revenue from the tail plus the \emph{expected welfare} from the core. We also note that the approach of Li and Yao eventually upper bounds the optimal revenue of the core with the expected welfare anyway.

\begin{lemma}[{Core-Tail}  Decomposition]\label{lem:easy}
$\Rev( D) \leq \Welfare(D_\emptyset^C) + \sum_A p_A \Rev(D_A^T)$
\end{lemma}
\begin{proof}
By Lemma~\ref{lem:HN1},
$$\Rev(D_A) \leq \Welfare(D_A^C) + \Rev(D_A^T)$$
for all $A$.  Also, since $\Welfare(D_A^C)$ is the expected sum of values for items \emph{not} in $A$, we have
$$\Welfare(D_A^C) \leq \Welfare(D_\emptyset^C).$$
By Lemma~\ref{lem:HN2},
\begin{align*}
\Rev(D) & \leq \sum_A p_A \Rev(D_A) \\
& \leq \sum_A p_A \left(\Welfare(D^C_A) + \Rev(D_A^T)\right) \\
& \leq \left(\sum_A p_A\right) \Welfare(D_\emptyset^C) + \sum_A p_A \Rev(D_A^T).
\end{align*}
As $\sum_A p_A = 1$ the desired result follows.
\end{proof}

\section{Main Result: Revenue Bounds for a Single Buyer}
\label{sec:one-ind}

In this section we focus on the case of a single buyer, $m=1$.  We will work toward proving our main result, which is that $\max\{\SRev(D), \BRev(D)\}$ is a constant-factor approximation to $\Rev(D)$ in this setting.  Our argument will make use of the {Core-Tail} decomposition, described in the previous section.  We will begin with a simpler result that illustrates our techniques: that $\Rev(D)$ is at most $(\ln n + 3)$ times $\SRev(D)$.  A logarithmic approximation was already established in~\cite{LiY13}; we obtain a slightly tighter bound, but the primary purpose of presenting this result is as a warm-up to introduce our techniques and those of~\cite{LiY13}.  We will then show how this bound can be improved to a constant by considering the maximum of $\SRev(D)$ and $\BRev(D)$.

\subsection{Warm-up: $(\ln n+ 3)\SRev \geq \Rev$}\label{subsec:SRev}
We first give a simple application of our approach to provide a bound on \SRev\ vs.\ \Rev, which is slightly improved relative to the bound obtained in~\cite{LiY13}.

\begin{theorem}\label{thm:SRev}
For a single buyer, and any $c \geq 1/n$, $(2+1/c + \ln c + \ln n) \SRev(D) \geq \Rev(D)$. This is minimized at $c = 1$, yielding $(\ln n + 3)\SRev(D) \geq \Rev(D)$.
\end{theorem}
The idea of the proof is to consider the {Core-Tail} decomposition of $D$, choosing $t_i = cn$ for each item $i$.  By the {Core-Tail}  Decomposition Lemma (Lemma \ref{lem:easy}), Theorem \ref{thm:SRev} follows if we can bound the optimal revenue from the tail and the expected welfare from the core, given this choice of $\{t_i\}_{i \in [n]}$.


We begin with Proposition \ref{prop:SRevTail}, which effectively shows that when $c$ is a constant, the revenue from the tail {(when $t_i = cn$ for each item $i$)} is at most a constant times $\SRev(D)$. The intuition behind this result is that each item $i$ lies in the tail with probability $p_i \leq 1/t_i = 1/cn$, and hence there will often be at most a single item whose value lies in the tail.  In this case, the revenue from the values in the tail is certainly no more than $\SRev(D)$, since the optimal mechanism can do no better than setting the optimal price for the single item present.  To bound the revenue contribution when many values lie in the tail, the relatively weak bound in Lemma~\ref{thm:HN} will suffice.


\begin{proposition}\label{prop:SRevTail}
For a single buyer, and any $c > 0$, if $t_i = cn$ for all $i$, then $\sum_A p_A \Rev(D_A^T) \leq (1+1/c)\SRev(D)$.
\end{proposition}

\begin{proof}
By Lemma~\ref{thm:HN} and Lemma~\ref{lem:LY2}, $\Rev(D_A^T) \leq  |A| \SRev(D_A^T) \leq \sum_{i \in A} |A| r_i/p_i$. Therefore, we may rewrite the sum by first summing over item $i$, and then summing over every set $A$ containing $i$, obtaining:
$$\sum_A p_A \Rev(D_A^T) \leq  \sum_A p_A \sum_{i \in A} |A| r_i/p_i =\sum_i  r_i\sum_{A \ni i} |A| \cdot p_A /p_i.$$
We now wish to interpret the term $\sum_{A \ni i} |A| \cdot p_A/p_i$. Observe that $p_A/p_i$ is exactly the probability that the set $A$ of items are in the tail {and all other items are not},
conditioned on $i$ being in the tail, and $|A|$ is just the size of $A$. Summing over all $A \ni i$ therefore yields the expected size of the set of items in the tail, conditioned on $i$ being in the tail.\footnote{This observation is due to Aviad Rubinstein, and we thank him for allowing us to include it.  An earlier version of this paper presented a $(\ln n + 5)$-approximation in Theorem~\ref{thm:SRev} and a $7.5$-approximation in Theorem~\ref{thm:main}. This observation improved those factors to $(\ln n + 3)$ and 6, respectively.} Clearly, {as $i$ is in $A$,} this expectation is just $1+\sum_{j \neq i} p_j$, which is at most $1+1/c$ by Lemma~\ref{lem:LY1}. As we have just observed that $\sum_{A \ni i} |A| p_A/p_i \leq 1+1/c$. Thus, we have now shown that $\sum_A p_A \Rev(D_A^T) \leq \sum_i (1+1/c)r_i$, which is exactly $(1+1/c)\SRev(D)$.
\end{proof}

Having established a bound on the revenue of the tail, we turn to the welfare of the core.  For this, we use the definition of $r_i = \SRev(D_i)$ to directly bound $\Pr_{v_i \leftarrow D_i}[v_i > x]$ for all $x$, and then take an expectation over the range of the core.

\begin{proposition}\label{prop:SRevCore}
For a single buyer, and any $c \geq 1/n$, if $t_i = cn$ for all $i$, then $(1+\ln c + \ln n) \SRev(D) \geq \Welfare(D_\emptyset^C)$.
\end{proposition}
\begin{proof}
Note that $\Welfare(D_\emptyset^C) = \sum_i \Welfare(D_i^C) \leq \sum_i  \int_0^{cnr_i} \Pr_{v_i \leftarrow D_i}[v_i > x]dx$. The last inequality would be equality if we replaced $v_i$ with a random variable drawn from $D_i^C$, but since $v_i$ stochastically dominates such a random variable, we get an inequality instead. As the optimal revenue of $D_i$ is $r_i$, this means that $\Pr_{v_i\leftarrow D_i}[v_i > x] \leq \min\{1,r_i/x\}$. So we have
\begin{align*}
\Welfare(D_i^C) & \leq \int_0^{r_i} dx + \int_{r_i}^{cnr_i} (r_i/x) dx
= r_i + r_i (\ln (cnr_i) - \ln (r_i) )
= r_i (1+\ln n + \ln c)
\end{align*}
As $c \geq 1/n$, the above breakdown of the integral is valid, and summing this guarantee over all $i$ yields the proposition.
\end{proof}

Combining Propositions \ref{prop:SRevTail} and \ref{prop:SRevCore} with Lemma \ref{lem:easy} yields Theorem \ref{thm:SRev}.

\subsection{Main Result: $6 \cdot \max\{\SRev, \BRev\} \geq \Rev$}
In this section we prove our main result, showing that the best of selling items separately and bundling all of them together is a constant-factor approximation to the optimal mechanism. The proof will follow a similar outline to that of Section~\ref{subsec:SRev}, proving propositions similar to Propositions~\ref{prop:SRevTail} and~\ref{prop:SRevCore}. The notable difference is that we will need to be more careful in defining the core.

When all $D_i$ are identical, the approach in Section~\ref{subsec:SRev} (setting each $t_i = cn$) can be leveraged to yield the bound $O(1) \cdot \BRev \geq  \Rev$ (\cite{LiY13}), but fails in the case that a small number $k$ of items contributes the majority of the optimal revenue.  To see the problem, note that the previous definition of the cutoffs $t_i$ depends on the number of items $n$, {but the number of items can be made arbitrarily large while keeping the problem essentially the same (by adding extra items whose values are deterministically $0$).}  The effect is that analysis of the same underlying instance changes as a result of these dummy items, so one should not expect approximations independent of $n$. In particular, the cutoffs $t_i$ will be larger than necessary when value distributions are asymmetric, and a few items contribute most of the revenue.  

Ideally, our analysis would be invariant under addition of dummy items. To accomplish this, we let $t_i$ scale inverse proportionally to $r_i$, so that high-revenue items are more likely to occur in the tail.  This allows us to capture scenarios in which revenue comes primarily from one heavy item (by analyzing the tail), as well as instances driven by the combined contribution of many light items (by analyzing the core).  Indeed, note that if we set $t_i = c r / r_i$, then the boundary between core and tail becomes $t_i r_i = c r = c \SRev(D)$ for each item (so while the choice of $t_i$ is non-uniform, the absolute cutoffs $t_i r_i$ are uniform).  This turns out to be precisely the threshold that we need to attain constant-factor approximation bounds for both the core and the tail, simultaneously. {We are now ready to state and prove our main result. }


\begin{theorem}
\label{thm:main}
For $m=1$ buyer and $n$ items, $\Rev(D) \leq 6\max\{\SRev(D), \BRev(D)\}$.
\end{theorem}

As in Theorem \ref{thm:SRev}, our approach will be to apply the {Core-Tail}  Decomposition Lemma (Lemma \ref{lem:easy}) with an appropriate choice of values $t_i$, then bound separately the revenue from the tail and the welfare from the core.  


\begin{proposition}\label{prop:maxTail}
For a single buyer, when $t_i = r/r_i$ for each $i$, $\sum_A p_A \Rev(D_A^T) \leq 2\SRev(D)$.
\end{proposition}

\begin{proof}
We begin similarly to the proof of Proposition~\ref{prop:SRevTail}, using Lemma~\ref{thm:HN} and Lemma~\ref{lem:LY2} to write $\Rev(D_A^T) \leq  |A| \SRev(D_A^T) \leq \sum_{i \in A} |A| r_i/p_i$. Again, summing this over all $A$ yields:

$$\sum_A p_A \Rev(D_A^T) \leq  \sum_i \sum_{A \ni i} |A| p_A r_i /p_i.$$

Just like in Proposition~\ref{prop:SRevTail}, $\sum_{A \ni i} |A| p_A/p_i$ is exactly the expected number of items in the tail, conditioned on $i$ being in the tail. It's again clear that this sum is exactly $1 + \sum_{j \neq i} p_j$. By Lemma~\ref{lem:LY1}, this is at most $1 + \sum_{j \neq i} 1/t_j$. By our choice of $t_i$, the second term is upper bounded by $1$, as $t_j = r/r_j$ and $\sum_j r_j = r$. Therefore, $\sum_{A \ni i} |A| p_A /p_i \leq 2$, and $\sum_A p_A \Rev(D_A^T) \leq 2\SRev(D)$.
\end{proof}

We now turn to bounding the welfare from the core.  We will use the small range of the core to derive an upper bound on the variance of its welfare.  This will allow us to conclude that the welfare is highly concentrated whenever it is sufficiently large relative to $\SRev(D)$.  Thus, if the welfare is ``small'' compared to $\SRev(D)$, then selling separately extracts most of the welfare (within the core); otherwise the welfare concentrates and so bundling extracts most of the welfare (within the core).  The following lemma of~\cite{LiY13} will be helpful for this approach; its proof appears in Appendix \ref{app:one-bidder-proofs} for completeness.

\begin{lemma}\label{lem:var}(\cite{LiY13})
Let $F$ be a one-dimensional distribution with optimal revenue at most $y$ supported on $[0,ty]$. Then $\Var(F) \leq (2t-1)y^2$.
\end{lemma}
%

\begin{corollary}\label{cor:var} For a single buyer, and any choice of $t_i$, $\Var(D_i^C) \leq 2t_ir^2_i$.
\end{corollary}
\begin{proof}
$\Rev(D_i^C) \leq r_i$, and the distribution $D_i^C$ is supported on $[0,t_i r_i]$. Therefore, plugging into Lemma~\ref{lem:var} (and relaxing) yields the desired bound.
\end{proof}

\begin{proposition}\label{prop:maxCore}
For a single buyer, when $t_i = r/r_i$ for every $i$, it holds that \\
 $\max\{\SRev(D), \BRev(D)\} \geq \frac{1}{4} \Welfare(D_\emptyset^C).$
\end{proposition}

\begin{proof}
There are two cases to consider. If $\Welfare(D_\emptyset^C) \leq 4 r$, then we trivially have that $\SRev(D) = r \geq \frac{1}{4} \Welfare(D_\emptyset^C)$ as required.

On the other hand, if $\Welfare(D_\emptyset^C) \geq 4 r$, then Corollary~\ref{cor:var} tells us that $\Var(D_i^C) \leq 2t_ir_i^2$. Summing over all $i$ and recalling that $t_i = r/r_i$ we get
\begin{align*}
\Var(D_\emptyset^C) & = \sum_i \Var(D_i^C)
\leq 2\sum_i t_i r_i^2
= 2r^2.
\end{align*}
So $\Var(D_\emptyset^C) \leq 2r^2$ and $\Welfare(D_\emptyset^C) \geq 4 r$. By Chebyshev's inequality, we get
\begin{align*}
\Pr_{\vec{v}\leftarrow D}\left[\sum_i v_i \leq \frac{2}{5}\cdot \Welfare(D_\emptyset^C)\right] & \leq  \frac{2r^2}{\left(1 - \frac{2}{5}\right)^2 \cdot \Welfare(D_\emptyset^C)^2} \\ &
\leq \frac{25r^2}{72r^2} = \frac{25}{72}.
\end{align*}
Since $\BRev(D)$ is at least the revenue obtained by setting price $\frac{2}{5} \cdot \Welfare(D_\emptyset^C)$ on the grand bundle, $\BRev(D) \geq (\frac{2}{5} \cdot \Welfare(D_\emptyset^C)) \cdot \frac{47}{72} = \frac{47}{180}\cdot \Welfare(D_\emptyset^C)$.  As $\frac{47}{180} > \frac{1}{4}$, $\BRev(D) > \frac{1}{4}\Welfare(D_\emptyset^C)$ as required.
\end{proof}

Combining Propositions \ref{prop:maxTail} and \ref{prop:maxCore} with Lemma \ref{lem:easy} yields Theorem \ref{thm:main}. Our analysis was improved to provide a bound of 5.2 in~\cite{MaS15}.~\cite{Rubinstein16} provides a construction $D$ (with $m =1$) such that $\max\{\SRev(D),\BRev(D)\} = (1/2+o(1))\cdot \Rev(D)$ (proving the analysis cannot be improved beyond a bound of $2$). It is an interesting open question to further narrow the gap between $2$ and $5.2$.


\section{Revenue Bounds for Multiple Buyers}
\label{sec:multi}

Here we extend our results to multiple buyers with all valuations sampled independently (again, not necessarily identically). We first show in Theorem~\ref{thm:manySRev} that selling items separately achieves a logarithmic (in the number of items $n$) approximation to the optimal revenue. In Section~\ref{sec:multiconcentration} (Theorem~\ref{thm:manyMax}), we explore the limits of our techniques in the multi-buyer case. Specifically, we establish that, like in the single buyer case, the only case in which selling items separately fails to achieve a good approximation is if welfare is highly concentrated. Unfortunately, such concentration is no longer sufficient to achieve a constant approximation by bundling all items together. This is so because even though the welfare is concentrated, the allocation of items to buyers which provides such welfare can change dramatically between realizations. Indeed, in Proposition~\ref{prop:lb-many-iid} we show not only that $\BRev(D)$ fails to provide a constant approximation to the optimal mechanism, but even $\PRev(D)$ fails, and this is so even when item values are sampled i.i.d. for all items and buyers.

Finally, we explore the connection between $\SRev(D),\BRev(D)$, and $\PRev(D)$ for multiple buyers. In Section~\ref{sec:multiiid} we establish that when \emph{either buyers or items} (not necessarily both) are i.i.d., then $\max\{\SRev(D),\BRev(D)\}$ achieves a constant-factor approximation to $\PRev(D)$ (Theorems~\ref{thm:previid} and~\ref{thm:previiditems}). We also establish that this approximation guarantee fails when neither buyers nor items are i.i.d. (Proposition~\ref{prop:lb-prev-max}) --- $\max\{\SRev(D),\BRev(D)\}$ can guarantee at best an $\Omega(\ln(n))$-approximation to $\PRev(D)$ in general (which is already achieved by $\SRev(D)$ itself, even when compared to $\Rev(D)$). 

Together, these provide a complete picture of the gaps between these three quantities (see {Table \ref{tab:one-cor} in Appendix~\ref{app:tables} for a summary of these results}).


\subsection{Extension: $(\ln n+6) \SRev \geq \Rev$}\label{sec:multisrev}

We first show that selling items separately achieves a logarithmic (in the number of items, $n$) approximation to the optimal revenue.

\begin{theorem}\label{thm:manySRev}
For any number $m$ {of buyers and $n$ items,} $(2+2e^{1/4} + \ln 4 + \ln n) \SRev(D) \geq \Rev(D)$. (Note that $2+2e^{1/4} + \ln 4 < 6$.)
\end{theorem}



Our proof will proceed via {amplification}.  We will begin with the (weak) bound on $\SRev$ vs. $\Rev$ from Lemma \ref{thm:HN}, then show in Theorem \ref{thm:amplification} how to amplify any such bound into an improved bound.  We will then iterate this amplification process over and over, until we reach the desired logarithmic approximation (which will be a fixed point of the amplification process).  To prove the amplification theorem, we use an approach similar to the single-buyer analysis from Section \ref{subsec:SRev}.  That is, we will apply the {Core-Tail}  Decomposition Lemma (Lemma \ref{lem:easy}), then bound the revenue of the tail and the welfare of the core with respect to $\SRev(D)$. {The first step in the proof of Theorem \ref{thm:manySRev} is the following amplification theorem, where any current bound (in terms of $a$) is improved.}


\begin{theorem}[Amplification]\label{thm:amplification}
Assume that for some $a,m > 0$, it holds that for any number $n$ of items and all $D$ on $m$ buyers and $n$ items that $a\cdot n\cdot \SRev(D) \geq \Rev(D)$. Then, for any $c \geq 1/a$, and all $D$ on $m$ buyers and $n$ items, $(2+2e^{1/ca}/c + \ln c + \ln a + \ln n) \SRev(D) \geq \Rev(D)$ as well. In particular, when $a \geq 1$, setting $c = 1$ yields $(2+2e^{1/a} + \ln a + \ln n) \SRev(D) \geq \Rev(D)$.
\end{theorem}

To prove Theorem \ref{thm:amplification}, we will apply the Core-Tail Decomposition Lemma (Lemma \ref{lem:easy}), using $t_i = c\cdot a\cdot n$ for each $i$.  Theorem \ref{thm:amplification} will then follow from bounds on the revenue from the tail and the expected welfare from the core, which we establish in the two following propositions.

\begin{proposition}\label{manyTail} Assume that for some $a,m > 0$, it holds that for any number $n$ of items and all $D$ on $m$ buyers and $n$ items that $a\cdot n\cdot \SRev(D) \geq \Rev(D)$. Then for all $D$ on $m$ buyers and $n$ items and $c > 0$, setting $t_i = c\cdot a\cdot n$ for all $i$ implies $\sum_A p_A \Rev(D_A^T) \leq (1+2e^{1/ca}/c) \SRev(D)$.
\end{proposition}

\begin{proof}
The following proof is similar to that of Proposition~\ref{prop:SRevTail}, with two differences. First, we start with the bound $\Rev(D_A^T) \leq a|A|\SRev(D_A^T)$ (since that is our starting hypothesis, instead of $\Rev(D_A^T) \leq |A|\SRev(D_A^T)$ as in the single-buyer case). Second, we have to make use of the fact that when there is only one item, $\SRev(D_A^T)=\Rev(D_A^T)$ and use this tighter bound whenever $|A| = 1$. We continue now with the proof.

By hypothesis and Lemma~\ref{lem:LY2}, for $A$ with $|A|>1$ it holds that $\Rev(D_A^T) \leq  a|A| \SRev(D_A^T) \leq \sum_{i \in A} a|A| r_i/p_i$. Combining with $\SRev(D_A^T)=\Rev(D_A^T)$ for the case that $|A|=1$, we can rewrite
$$\sum_A p_A \Rev(D_A^T) \leq  \sum_i\left( r_i + \sum_{j = 2}^n aj \sum_{A \ni i,|A| = j} p_A r_i /p_i\right)$$
Observe that $p_A = (\prod_{i \in A} p_i)(\prod_{i \notin A} (1-p_i))\leq \prod_{i \in A} p_i$ and thus
$p_A r_i /p_i \leq \prod_{k \in A - \{i\}} p_k r_i$.
We then have that
$$\sum_{A \ni i,  |A| = j} p_A r_i /p_i \leq r_i \sum_{A \ni i, |A| = j} \prod_{k \in A - \{i\}} p_k.$$
Furthermore, by Lemma~\ref{lem:LY1}, we have that each $p_k \leq 1/(c\cdot a\cdot n)$, so we have
\begin{align*}
\sum_{j=2}^n \sum_{A \ni i, |A| = j} aj \prod_{k \in A - \{i\}} p_k & \leq \sum_{j=2}^{n}a j\binom{n-1}{j-1}/(c\cdot a\cdot n)^{j-1}\\
&\leq \sum_{j=2}^{n}  \frac{j (n-1)^{j-1}}{(j-1)! c^{j-1} a^{j-2}n^{j-1}}\\
&\leq \sum_{j=2}^n \frac{2(j-1)}{(j-1)!c^{j-1}a^{j-2}}\\
&\leq \sum_{j=2}^n \frac{2}{(j-2)!c^{j-1}a^{j-2}}\\
&\leq \frac{2e^{1/ca}}{c}
\end{align*}
The last inequality makes use of the fact that $\sum_{j=0}^\infty \frac{1}{j!c^ja^j}$ is the Taylor expansion for $e^{x/ca}$ evaluated at $x = 1$. Adding back the $j=1$ term that we handled outside the sum (making use of the fact that $\SRev = \Rev$ on single-item distributions) and 
summing over all $i$ of $r_i$ times the above inequality yields the proposition.
\end{proof}

The following bound on the welfare from the core follows similarly to Proposition \ref{prop:SRevCore}. 

\begin{proposition}\label{manyCore}
For any number of buyers and any {positive} $a$ and $c$ with $a\cdot c \geq 1/n$, if $t_i = c\cdot a\cdot n$ for all $i$, then $(1+\ln c +\ln a+ \ln n) \SRev(D) \geq  \Welfare(D_\emptyset^C)$.
\end{proposition}
\begin{proof}
Note that $\Welfare(D_\emptyset^C) = \sum_i \Welfare(D_i^C) \leq \sum_i \int_0^{canr_i} Pr[v^*_i > x]dx$ (recall that $v^*_i:=\max_j \{v_{ij}\}$). The last inequality would be equality if we replaced $v^*_i$ with a random variable that is the maximum value in a sample drawn from $D_i^C$, but since $v^*_i$ stochastically dominates such a random variable, we get an inequality instead. As the optimal revenue of $D_i$ is $r_i$, this means that $Pr[v^*_i > x] \leq \min\{1,r_i/x\}$. So we have
\begin{align*}
\Welfare(D_i^C) & \leq \int_0^{r_i} dx + \int_{r_i}^{canr_i} (r_i/x) dx\\
&= r_i + r_i (\ln (c\cdot a\cdot n\cdot r_i) - \ln (r_i) )\\
&= r_i (1+\ln n + \ln c + \ln a)
\end{align*}
Summing this bound over all $i$ yields the proposition.
\end{proof}

Theorem \ref{thm:amplification} then follows from Propositions \ref{manyTail} and \ref{manyCore}, together with Lemma \ref{lem:easy}.  We now show how to prove Theorem \ref{thm:manySRev} using Theorem \ref{thm:amplification}.

\begin{proof}[Proof of Theorem \ref{thm:manySRev}]
The goal is to iteratively apply Theorem~\ref{thm:amplification} starting with $a = m$ (which is a valid hypothesis, by Lemma~\ref{thm:HN}), until we can apply it once with $a \leq 4$. 

So let us start with an application of Theorem~\ref{thm:amplification} from $a=m$. This yields a bound of the form $a'n \SRev(D) \geq \Rev(D)$ for some new $a'$. First, perhaps already $a' \leq 4$. If not, we can then apply Theorem~\ref{thm:amplification} again, taking $a$ to be this new value $a'$.  We can iteratively apply Theorem~\ref{thm:amplification} until we reach $a \leq 4$. One can verify that, for all $n \geq 2$, $a \geq 4$, the function $f(a) = (2 + 2e^{1/a} + \ln a + \ln n)/n$ satisfies $f(a) < a - 1$.  Therefore, $mn-4$ iterations suffice to get $a \leq 4$. Once the hypothesis holds with some $a \leq 4$, we can apply Theorem~\ref{thm:amplification} one final time to conclude Theorem~\ref{thm:manySRev}.
\end{proof}


\subsection{Comparing $\SRev(D),\BRev(D),\PRev(D)$ for multiple buyers}\label{sec:multiiid}
In this section, we investigate the relationship between $\max\{\SRev(D),\BRev(D)\}$ and $\PRev(D)$ for multiple buyer with independent items. The main results of this section establish that when \emph{either} buyers \emph{or} items are i.i.d. that the gap is at most a constant factor. Proposition~\ref{prop:lb-prev-max} establishes that this gap can be $\Omega(\ln n)$ when neither buyers nor items are i.i.d. (which is tight, as $\max\{\SRev(D),\BRev(D)\} \geq \SRev(D) \geq \Omega(1/\ln(n)) \cdot \Rev(D) \geq \PRev(D)$ by Theorem~\ref{thm:manySRev}).
\begin{theorem}\label{thm:previid} Let $D$ have any number of items and any number of i.i.d. buyers (that is, $D^j = D^{j'}$ for all buyers $j,j'$, but perhaps $D_i \neq D_{i'}$ for some $i,i'$). Then $\PRev(D) \leq O(\max\{\SRev(D),\BRev(D)\})$.
\end{theorem}

\begin{theorem}\label{thm:previiditems} Let $D$ have any number of i.i.d. items and any number of buyers (that is, $D_i= D_{i'}$ for all items $i,i'$, but perhaps $D^j \neq D^{j'}$ for some $j,j'$). Then $\PRev(D) \leq O(\max\{\SRev(D),\BRev(D)\})$.
\end{theorem}

The intuition for Theorems~\ref{thm:previid} and~\ref{thm:previiditems} is the following. Consider any partition mechanism which partitions the items into $S_1,\ldots, S_k$, and refer by $D_{S_i}$ to the distribution $D$ restricted to only items in $S_i$ (and therefore the partition mechanism achieves revenue $\sum_i \BRev(D_{S_i})$). We say that $S_i$ is \emph{separable} if $\SRev(D_{S_i}) = \Omega(\BRev(D_{S_i}))$. If $\mathcal{S}$ denotes the set  of all indices $i$ for which $S_i$ is separable, it then immediately follows that $\SRev(D) = \Omega(\sum_{i \in \mathcal{S}} \BRev(D_{S_i}))$ (Lemma~\ref{lem:separable}). Similarly, we say that $S_i$ is \emph{bundlable} for $j$ if 
selling $S_i$ as a bundle \emph{only to buyer $j$}, ignoring all other buyers, generates revenue $\Omega(\BRev(D_{S_i}))$. We then argue that if $\mathcal{B}_j$ denotes the set  of all indexes $i$ for which $S_i$ is bundlable for buyer $j$, that $\BRev(D) = \Omega(\sum_{i \in \mathcal{B}
_j} \BRev(D_{S_i}))$ (Lemma~\ref{lem:bundlable}). This step is not quite as trivial as Lemma~\ref{lem:separable}, but still fairly simple. The most interesting step of the proof for both theorems is showing that \emph{when either buyers or items are i.i.d.} there is an approximately-optimal partition mechanism and buyer $j$ such that \emph{every set $S_i$ that is not separable, is bundlable for that buyer $j$}.
The key step is stated formally in Propositions~\ref{prop:good} and~\ref{prop:seporbund}, and full details for all proofs are in Appendix~\ref{app:many-bidder-proofs}.

Proposition~\ref{prop:lb-prev-max} below establishes, however, that at least one of the i.i.d.~assumptions used in Theorems~\ref{thm:previid} and~\ref{thm:previiditems} is necessary.

\begin{definition}
We denote by $\er_k$ the Equal Revenue curve truncated at $k$: the single-dimensional distribution with $F(x) = 1-1/x$ for all $x \in [1,k]$, $F(x) = 0$ for all $x < 1$, and $F(x) = 1$ for all $x \geq k$ (i.e. it is an equal revenue curve with all mass above $k$ moved to a point mass at $k$).
\end{definition}

\begin{proposition}
\label{prop:lb-prev-max}
Let $D$ have $n$ items and $m = \sqrt{n}$ buyers. Partition the items into $\sqrt{n}$ disjoint sets of size $\sqrt{n}$, $S_1,\ldots, S_m$. Let buyer $j$ have value $0$ for every item not in $S_j$, and value independently drawn from $\er_{n^{1/8}}$ for each item in $S_j$. Then $\max\{\SRev(D), \BRev(D)\}\leq \PRev(D)/\Omega(\log n)$.
\end{proposition}

A full proof of Proposition~\ref{prop:lb-prev-max} appears in Appendix~\ref{app:many-bidder-proofs}. The high level idea is that each buyer $j$ is only interested in items in $S_j$, which are disjoint. So we should partition the items and run the optimal single-buyer auction within each (which would generate revenue $\Omega(n \ln n)$). Selling instead each partition separately to the intended buyer generates revenue only $O(n)$, and selling the entire bundle together causes $n-\sqrt{n}$ items to be wasted, and generates revenue only $O(\sqrt{n} \ln (n))$. To help process this example in the language of our proof outline of Theorems~\ref{thm:previid} and~\ref{thm:previiditems}, observe that no $S_i$ is separable (because $\SRev(D_{S_i}) = \sqrt{n}$, while $\BRev(D_{S_i}) = \Theta(\sqrt{n} \cdot \ln (n))$). Also, observe that \emph{only} $S_j$ is bundlable for $j$ (because $j$ has value $0$ for all items not in $S_j$). The key step (Proposition~\ref{prop:seporbund}) towards Theorems~\ref{thm:previid} and~\ref{thm:previiditems} states that this phenomenon cannot occur in an example with either i.i.d.~items or i.i.d.~buyers.

\subsection{A Lower Bound: $\PRev\leq \Rev/\Omega(\log n)$ even for i.i.d. Item Values}

We next show that there is a setting with many buyers with item valuations that are sampled i.i.d from the same distribution, for which $\PRev(D)$ (and thus also $\max\{\SRev(D), \BRev(D)\}$) provides a poor approximation to $\Rev(D)$. Intuitively, the key feature our example possesses is that for any \emph{fixed} set $S$ of $\Theta(\sqrt{n})$ items, it is extremely unlikely that any buyer values $S$ particularly highly. Yet, for all buyers, it is extremely \emph{likely} that they value \emph{some} set of $\Theta(\sqrt{n})$ items highly. The former property allows us to claim that all partition mechanisms perform poorly, while the latter property allows us to design a posted-price mechanism which performs well. The construction will make use of the following distribution:


\begin{proposition}
\label{prop:lb-many-iid}
When $D$ is such that each of $m=\sqrt{n}$ buyers have i.i.d. values for each of $n$ items drawn from a distribution that is a point-mass at $0$ with probability $1-1/\sqrt{n}$ and drawn from $\er_{n^{1/8}}$ with the remaining probability: $\PRev(D)\leq \Rev(D)/\Omega(\log n)$.
\end{proposition}

A full proof of Proposition~\ref{prop:lb-many-iid} appears in Appendix~\ref{app:many-bidder-proofs}. The high-level idea is that only one buyer in expectation has non-zero value for each item, but each buyer has non-zero value for $\sqrt{n}$ items in expectation. As a result, selling the entire grand bundle at once achieves poor revenue ($O(\sqrt{n}\ln(n))$), since $n-\sqrt{n}$ are likely to be unvalued by the winner. Similarly, selling separately is suboptimal ($O(n)$) because each buyer's value for the $\sqrt{n}$ items they like concentrates around its expectation, and selling separately doesn't exploit this. Instead, a posted-price mechanism which allows each buyer to pick any subset of $\Theta(\sqrt{n})$ remaining items for $\Theta(\sqrt{n} \ln(n))$ has the property that with high probability, each buyer wishes to purchase a set of remaining items, generating revenue $\Omega(n \ln(n))$. Intersetingly, note that, because $\SRev(D)$ is a $(\ln(n)+6)$-approximation to $\Rev(D)$, and we have just claimed that this posted-price mechanism achieves a $\Theta(\ln(n))$-factor more revenue than $\PRev(D)$ (and therefore $\SRev(D)$), this posted-price mechanism must be a constant-factor approximation.



\subsection{A Concentration Result}\label{sec:multiconcentration}

Finally, we explore the limits of our single-buyer approach for multiple buyers. Specifically, we establish sufficient conditions for $\SRev(D)$ to be a constant-factor approximation to $\Rev(D)$ with multiple buyers.  We will show (Theorem \ref{thm:manyMax}) that this occurs unless the welfare of $D$ is sufficiently well concentrated around its expectation. Proposition~\ref{prop:lb-many-iid} establishes that this concentration \emph{does not suffice} even for $\PRev(D)$ to guarantee a constant-factor approximation, so this result identifies that the main challenge in extending our work to multiple buyers is leveraging concentration of welfare to get a constant-factor approximation (and~\cite{Yao15} accomplishes this via an entry fee).

We begin with a corollary of Theorem \ref{thm:manySRev}, which will be useful for our analysis.
\begin{corollary}
\label{cor:many}
For any number of buyers and $n$ items, $4\cdot n \cdot \SRev(D) \geq \Rev(D)$.
\end{corollary}
\begin{proof}
This is a direct application of Theorem~\ref{thm:manySRev} and noting that $6 + \ln n \leq 4n$ for all $n \geq 2$.
\end{proof}

We next prove an alternative bound on the revenue from the tail of the distribution $D$, using a familiar choice of $t_i$.  

\begin{proposition}\label{prop:manymaxTail} For any number of buyers, if $t_i = 4r/r_i$ for all $i$, then $\sum_A p_A \Rev(D_A^T) \leq 5 \SRev(D)$.
\end{proposition}

\begin{proof}
Again, we begin by rewriting $\sum_A p_A \Rev(D_A^T)$ using Corollary~\ref{cor:many}, and reordering:
$$\sum_A p_A \Rev(D_A^T) \leq \sum_A 4p_A|A|\sum_{i \in A} r_i/p_i = \sum_i r_i \sum_{A \ni i}4 |A|p_A/p_i = 4 \sum_i r_i \sum_{A \ni i} |A| p_A/p_i.$$
Again, the value $\sum_{A \ni i} |A| p_A/p_i$ is exactly the expected number of items in the tail, conditioned on $i$ being in the tail. Therefore, $\sum_{A \ni i} |A| p_A/p_i \leq 1 + \sum_{j \neq i} p_j \leq 1+1/4$ by Lemma~\ref{lem:LY1}. Therefore, $\sum_A p_A \Rev (D_A^T) \leq 5 \sum_i r_i = 5 \SRev$.
\end{proof}

We are now ready to establish the claimed bound between $\SRev$ and $\Rev$, subject to the welfare of $D$ not being too concentrated around its expectation.

\begin{definition} We say that a one-dimensional distribution $F$ is $d$-concentrated if there exists a value $C$ such that $Pr_{x \sim F}[|x - C| \leq C/2] \geq d$.
\end{definition}

\begin{theorem}\label{thm:manyMax}
For any number of buyers, and any $c \geq 4\sqrt{2}$, either $(c+5) \SRev(D) \geq \Rev(D)$ or the welfare of $D$ (the random variable with expectation $\Welfare(D)$) is $(3/4 - \frac{24}{c^2})$-concentrated.
\end{theorem}

\begin{proof}
Let all $t_i = 4r/r_i$. Then combining Proposition~\ref{prop:manymaxTail} and Lemma~\ref{lem:easy} yields
$$5 \cdot \SRev(D) + \Welfare(D_\emptyset^C) \geq \Rev(D).$$
There are two cases to consider. First, perhaps $c\cdot \SRev(D) \geq \Welfare(D_\emptyset^C)$. In this case, we have $(c+5)\SRev(D) \geq \Rev(D)$ and the claim follows.

In the other case, $\Welfare(D_\emptyset^C) \geq c\cdot \SRev(D)$. In this case, Corollary~\ref{cor:var} tells us that $\Var(D_i^C) \leq 2t_ir_i^2$. Summing over all $i$ and recalling that $t_i = 4r/r_i$, we get
\begin{align*}
\Var(D_\emptyset^C) & \leq 2\sum_i t_i r_i^2
= 2\sum_i (4r) r_i
= 8r^2.
\end{align*}
So $\Var(D_\emptyset^C) \leq 8r^2$ and $\Welfare(D_\emptyset^C) \geq c r$. By Chebyshev's inequality, we get
\begin{align*}
Pr\left[\left|\sum_i v^*_i - \Welfare(D_\emptyset^C)\right| \geq \Welfare(D_\emptyset^C)/2\right] & \leq  \frac{8r^2}{\Welfare(D_\emptyset^C)^2/4}
\leq \frac{32r^2}{c^2r^2} = \frac{32}{c^2}
\end{align*}
meaning that the welfare of $D_\emptyset^C$ is $(1-\frac{32}{c^2})$-concentrated. The last step is observing that $\vec{v}$ is sampled in the support of $D_\emptyset^C$ with probability exactly $\prod_i (1-p_i)$. As $\sum_i p_i \leq \sum_i 1/t_i \leq \sum_i \frac{r_i}{4r}= 1/4$ and each $p_i \geq 0$, this is minimized when exactly one $p_i$ is $1/4$ and the rest are $0$, yielding $\prod_i (1-p_i) = 3/4$. So with probability at least $3/4$ $\vec{v}$ is in the support of $D_\emptyset^C$. When this happens, the welfare is $(1-\frac{32}{c^2})$ concentrated. So the welfare of $D$ is $(3/4 - \frac{24}{c^2})$-concentrated.
\end{proof}


\section{One Buyer with Correlated Values}
\label{sec:one-cor}
In this section, we study the relationship between \SRev(D), Max\{\SRev(D), \BRev(D)\}, and \PRev(D)\ for a single buyer with correlated values. Prior work of~\cite{BriestCKW15,HartN13} already shows that there is no hope of obtaining a non-zero bound between any of these quantities and \Rev(D) (because partition mechanisms are deterministic, and~\cite{BriestCKW15,HartN13} show that no deterministic mechanism achieves any non-zero approximation), even when there are only two items. But it is still important to understand the relationship between these mechanisms of varying complexity even if their revenue cannot compare to that of the optimal mechanism.
We show in Theorem~\ref{thm:one-cor-srev-brev} that for any correlated distribution $D$ for a single buyer and $n$ items, $\SRev(D)$ is a $O(\ln n)$ approximation to $\BRev(D)$, and thus also to Max\{\SRev(D), \BRev(D)\} and \PRev(D).\footnote{As \SRev\ approximates \BRev\ for any set of items, it can do so for any part in the partition in \PRev\ separately, and thus also approximates \PRev.}
We then show in Proposition~\ref{prop:lb-cor} that this bound is tight, there exists a distribution $D$ witnessing $\max\{\SRev(D), \BRev(D)\} \leq \PRev(D) / \Omega(\ln n)$. In other words, $\SRev(D)$ provides a logarithmic approximation to $\PRev(D)$, but taking $\max\{\SRev(D), \BRev(D)\}$ can't guarantee anything better (even for $m=1$ buyer). Both proofs appear in Appendix~\ref{app:one-cor}.
\begin{theorem}
\label{thm:one-cor-srev-brev}
For any distribution $D$ for a single buyer and $n$ items (arbitrarily correlated), $\BRev(D)\leq 5\ln(n) \SRev(D)$. Therefore, $\PRev(D) \leq 5 \ln (n) \SRev(D)$ as well.
\end{theorem}

\begin{proposition}
\label{prop:lb-cor}
There exists a (correlated) distribution $D$ of the valuation of a single buyer over $n$ items for which
$\max\{\SRev(D), \BRev(D)\}\leq \PRev(D)/\Omega(\ln n)$.
\end{proposition}


\section{Acknowledgments}
In an earlier version of this paper, we proved a factor of 7.5 in Theorem~\ref{thm:main}. This factor was later improved by Aviad Rubinstein to a factor of 6. We thank Aviad for allowing us to include this improvement in our paper.

\bibliographystyle{alpha}
\bibliography{MasterBib}

\newcommand{\etalchar}[1]{$^{#1}$}
\begin{thebibliography}{CGMW18}

\bibitem[AKW19]{AzarKW19}
Pablo~Daniel Azar, Robert Kleinberg, and S.~Matthew Weinberg.
\newblock Prior independent mechanisms via prophet inequalities with limited
  information.
\newblock {\em Games Econ. Behav.}, 118:511--532, 2019.

\bibitem[AM06]{Ausubel06}
Lawrence~M. Ausubel and Paul Milgrom.
\newblock The lovely but lonely vickrey auction.
\newblock In {\em Combinatorial Auctions, chapter 1}. MIT Press, 2006.

\bibitem[BCGZ18]{BabaioffCGZ18}
Moshe Babaioff, Yang Cai, Yannai~A. Gonczarowski, and Mingfei Zhao.
\newblock The best of both worlds: Asymptotically efficient mechanisms with a
  guarantee on the expected gains-from-trade.
\newblock In {\em Proceedings of the 2018 {ACM} Conference on Economics and
  Computation, Ithaca, NY, USA, June 18-22, 2018}, page 373, 2018.

\bibitem[BCKW15]{BriestCKW15}
Patrick Briest, Shuchi Chawla, Robert Kleinberg, and S.~Matthew Weinberg.
\newblock Pricing lotteries.
\newblock {\em J. Economic Theory}, 156:144--174, 2015.

\bibitem[BCWZ17]{BrustleCWZ17}
Johannes Brustle, Yang Cai, Fa~Wu, and Mingfei Zhao.
\newblock Approximating gains from trade in two-sided markets via simple
  mechanisms.
\newblock In {\em Proceedings of the 2017 {ACM} Conference on Economics and
  Computation, {EC} '17, Cambridge, MA, USA, June 26-30, 2017}, pages 589--590,
  2017.

\bibitem[BDHS15]{BateniDHS15}
MohammadHossein Bateni, Sina Dehghani, MohammadTaghi Hajiaghayi, and Saeed
  Seddighin.
\newblock Revenue maximization for selling multiple correlated items.
\newblock In {\em the 23rd Annual European Symposium on Algorithms (ESA)},
  2015.

\bibitem[BGN17]{BabaioffGN17}
Moshe Babaioff, Yannai~A. Gonczarowski, and Noam Nisan.
\newblock The menu-size complexity of revenue approximation.
\newblock In {\em Proceedings of the 49th Annual {ACM} {SIGACT} Symposium on
  Theory of Computing, {STOC} 2017, Montreal, QC, Canada, June 19-23, 2017},
  pages 869--877, 2017.

\bibitem[BK96]{BulowK96}
Jeremy Bulow and Paul Klemperer.
\newblock Auctions versus negotiations.
\newblock {\em The American Economic Review}, 86(1):180--194, 1996.

\bibitem[BSV16]{BalcanSV16}
Maria{-}Florina Balcan, Tuomas Sandholm, and Ellen Vitercik.
\newblock Sample complexity of automated mechanism design.
\newblock In {\em Advances in Neural Information Processing Systems 29: Annual
  Conference on Neural Information Processing Systems 2016, December 5-10,
  2016, Barcelona, Spain}, pages 2083--2091, 2016.

\bibitem[BSV18]{BalcanSV18}
Maria{-}Florina Balcan, Tuomas Sandholm, and Ellen Vitercik.
\newblock A general theory of sample complexity for multi-item profit
  maximization.
\newblock In {\em Proceedings of the 2018 {ACM} Conference on Economics and
  Computation, Ithaca, NY, USA, June 18-22, 2018}, pages 173--174, 2018.

\bibitem[BW19]{BeyhaghiW19}
Hedyeh Beyhaghi and S.~Matthew Weinberg.
\newblock Optimal (and benchmark-optimal) competition complexity for additive
  buyers over independent items.
\newblock In {\em Proceedings of the 51st ACM Symposium on Theory of Computing
  Conference (STOC)}, 2019.

\bibitem[Car17]{Carroll17}
Gabriel Carroll.
\newblock Robustness and separation in multidimensional screening.
\newblock {\em Econometrica}, 85(2):453--488, 2017.

\bibitem[CD11]{CaiD11}
Yang Cai and Constantinos Daskalakis.
\newblock Extreme-value theorems for optimal multidimensional pricing.
\newblock In {\em {IEEE} 52nd Annual Symposium on Foundations of Computer
  Science, {FOCS} 2011, Palm Springs, CA, USA, October 22-25, 2011}, pages
  522--531, 2011.

\bibitem[CD17]{CaiD17}
Yang Cai and Constantinos Daskalakis.
\newblock Learning multi-item auctions with (or without) samples.
\newblock In {\em 58th {IEEE} Annual Symposium on Foundations of Computer
  Science, {FOCS} 2017, Berkeley, CA, USA, October 15-17, 2017}, pages
  516--527, 2017.

\bibitem[CDO{\etalchar{+}}15]{ChenDOPSY15}
Xi~Chen, Ilias Diakonikolas, Anthi Orfanou, Dimitris Paparas, Xiaorui Sun, and
  Mihalis Yannakakis.
\newblock On the complexity of optimal lottery pricing and randomized
  mechanisms.
\newblock In {\em {IEEE} 56th Annual Symposium on Foundations of Computer
  Science, {FOCS} 2015, Berkeley, CA, USA, 17-20 October, 2015}, pages
  1464--1479, 2015.

\bibitem[CDP{\etalchar{+}}14]{ChenDPSY14}
Xi~Chen, Ilias Diakonikolas, Dimitris Paparas, Xiaorui Sun, and Mihalis
  Yannakakis.
\newblock The complexity of optimal multidimensional pricing.
\newblock In {\em Proceedings of the Twenty-Fifth Annual {ACM-SIAM} Symposium
  on Discrete Algorithms, {SODA} 2014, Portland, Oregon, USA, January 5-7,
  2014}, pages 1319--1328, 2014.

\bibitem[CDW16]{CaiDW16}
Yang Cai, Nikhil Devanur, and S.~Matthew Weinberg.
\newblock A duality based unified approach to bayesian mechanism design.
\newblock In {\em Proceedings of the 48th ACM Conference on Theory of
  Computation(STOC)}, 2016.

\bibitem[CEF{\etalchar{+}}18]{CaiEFLWZ18}
Yang Cai, Federico Echenique, Hu~Fu, Katrina Ligett, Adam Wierman, and Juba
  Ziani.
\newblock Third-party data providers ruin simple mechanisms.
\newblock {\em CoRR}, abs/1802.07407, 2018.

\bibitem[CGMW18]{ChengGMW18}
Yu~Cheng, Nick Gravin, Kamesh Munagala, and Kangning Wang.
\newblock A simple mechanism for a budget-constrained buyer.
\newblock In {\em the 14th International Workshop on Internet and Network
  Economics (WINE)}, 2018.

\bibitem[CH13]{CaiH13}
Yang Cai and Zhiyi Huang.
\newblock {Simple and Nearly Optimal Multi-Item Auctions}.
\newblock In {\em the 24th Annual ACM-SIAM Symposium on Discrete Algorithms
  (SODA)}, 2013.

\bibitem[CHK07]{ChawlaHK07}
Shuchi Chawla, Jason~D. Hartline, and Robert~D. Kleinberg.
\newblock {Algorithmic Pricing via Virtual Valuations}.
\newblock In {\em the 8th ACM Conference on Electronic Commerce (EC)}, 2007.

\bibitem[CHMS10]{ChawlaHMS10}
Shuchi Chawla, Jason~D. Hartline, David~L. Malec, and Balasubramanian Sivan.
\newblock {Multi-Parameter Mechanism Design and Sequential Posted Pricing}.
\newblock In {\em the 42nd ACM Symposium on Theory of Computing (STOC)}, 2010.

\bibitem[CLL17]{ChenLL17}
Jing Chen, Bo~Li, and Yingkai Li.
\newblock From bayesian to crowdsourced bayesian auctions.
\newblock {\em CoRR}, abs/1702.01416, 2017.

\bibitem[CM16]{ChawlaM16}
Shuchi Chawla and J.~Benjamin Miller.
\newblock Mechanism design for subadditive agents via an ex ante relaxation.
\newblock In {\em Proceedings of the 2016 {ACM} Conference on Economics and
  Computation, {EC} '16, Maastricht, The Netherlands, July 24-28, 2016}, pages
  579--596, 2016.

\bibitem[CMPY18]{ChenMPY18}
Xi~Chen, George Matikas, Dimitris Paparas, and Mihalis Yannakakis.
\newblock On the complexity of simple and optimal deterministic mechanisms for
  an additive buyer.
\newblock In {\em Proceedings of the Twenty-Ninth Annual {ACM-SIAM} Symposium
  on Discrete Algorithms, {SODA} 2018, New Orleans, LA, USA, January 7-10,
  2018}, pages 2036--2049, 2018.

\bibitem[CMS10]{ChawlaMS10}
Shuchi Chawla, David~L. Malec, and Balasubramanian Sivan.
\newblock {The Power of Randomness in Bayesian Optimal Mechanism Design}.
\newblock In {\em the 11th ACM Conference on Electronic Commerce (EC)}, 2010.

\bibitem[CMS15]{ChawlaMS15}
Shuchi Chawla, David~L. Malec, and Balasubramanian Sivan.
\newblock The power of randomness in bayesian optimal mechanism design.
\newblock {\em Games and Economic Behavior}, 91:297--317, 2015.

\bibitem[CR14]{ColeR14}
Richard Cole and Tim Roughgarden.
\newblock The sample complexity of revenue maximization.
\newblock In {\em Symposium on Theory of Computing, {STOC} 2014, New York, NY,
  USA, May 31 - June 03, 2014}, pages 243--252, 2014.

\bibitem[CZ17]{CaiZ17}
Yang Cai and Mingfei Zhao.
\newblock Simple mechanisms for subadditive buyers via duality.
\newblock In {\em Proceedings of the 49th Annual {ACM} {SIGACT} Symposium on
  Theory of Computing, {STOC} 2017, Montreal, QC, Canada, June 19-23, 2017},
  pages 170--183, 2017.

\bibitem[DDT14]{DaskalakisDT14}
Constantinos Daskalakis, Alan Deckelbaum, and Christos Tzamos.
\newblock {The Complexity of Optimal Mechanism Design}.
\newblock In {\em the 25th ACM-SIAM Symposium on Discrete Algorithms (SODA)},
  2014.

\bibitem[DDT17]{DaskalakisDT17}
Constantinos Daskalakis, Alan Deckelbaum, and Christos Tzamos.
\newblock Strong duality for a multiple-good monopolist.
\newblock {\em Econometrica}, 85(3):735--767, 2017.

\bibitem[DHP16]{DevanurHP16}
Nikhil~R. Devanur, Zhiyi Huang, and Christos{-}Alexandros Psomas.
\newblock The sample complexity of auctions with side information.
\newblock In {\em Proceedings of the 48th Annual {ACM} {SIGACT} Symposium on
  Theory of Computing, {STOC} 2016, Cambridge, MA, USA, June 18-21, 2016},
  pages 426--439, 2016.

\bibitem[DPT16]{DaskalakisPT16}
Constantinos Daskalakis, Christos~H. Papadimitriou, and Christos Tzamos.
\newblock Does information revelation improve revenue?
\newblock In {\em Proceedings of the 2016 {ACM} Conference on Economics and
  Computation, {EC} '16, Maastricht, The Netherlands, July 24-28, 2016}, pages
  233--250, 2016.

\bibitem[DRY15]{DhangwatnotaiRY15}
Peerapong Dhangwatnotai, Tim Roughgarden, and Qiqi Yan.
\newblock Revenue maximization with a single sample.
\newblock {\em Games and Economic Behavior}, 91:318--333, 2015.

\bibitem[EFF{\etalchar{+}}17a]{EdenFFTW17b}
Alon Eden, Michal Feldman, Ophir Friedler, Inbal Talgam{-}Cohen, and S.~Matthew
  Weinberg.
\newblock The competition complexity of auctions: {A} bulow-klemperer result
  for multi-dimensional bidders.
\newblock In {\em Proceedings of the 2017 {ACM} Conference on Economics and
  Computation, {EC} '17, Cambridge, MA, USA, June 26-30, 2017}, page 343, 2017.

\bibitem[EFF{\etalchar{+}}17b]{EdenFFTW17a}
Alon Eden, Michal Feldman, Ophir Friedler, Inbal Talgam{-}Cohen, and S.~Matthew
  Weinberg.
\newblock A simple and approximately optimal mechanism for a buyer with
  complements: Abstract.
\newblock In {\em Proceedings of the 2017 {ACM} Conference on Economics and
  Computation, {EC} '17, Cambridge, MA, USA, June 26-30, 2017}, page 323, 2017.

\bibitem[EH17]{ElmachtoubH17}
Adam~N. Elmachtoub and Michael~L. Hamilton.
\newblock The power of opaque products in pricing.
\newblock In {\em Web and Internet Economics - 13th International Conference
  (WINE)}, 2017.

\bibitem[FFR18]{FeldmanFR18}
Michal Feldman, Ophir Friedler, and Aviad Rubinstein.
\newblock 99{\%} revenue via enhanced competition.
\newblock In {\em Proceedings of the 2018 {ACM} Conference on Economics and
  Computation, Ithaca, NY, USA, June 18-22, 2018}, pages 443--460, 2018.

\bibitem[FGL15]{FeldmanGL15}
Michal Feldman, Nick Gravin, and Brendan Lucier.
\newblock Combinatorial auctions via posted prices.
\newblock In {\em Proceedings of the Twenty-Sixth Annual ACM-SIAM Symposium on
  Discrete Algorithms}, SODA '15, pages 123--135, Philadelphia, PA, USA, 2015.
  Society for Industrial and Applied Mathematics.

\bibitem[FLLT18]{FuLLT18}
Hu~Fu, Christopher Liaw, Pinyan Lu, and Zhihao~Gavin Tang.
\newblock The value of information concealment.
\newblock In {\em Proceedings of the Twenty-Ninth Annual {ACM-SIAM} Symposium
  on Discrete Algorithms, {SODA} 2018, New Orleans, LA, USA, January 7-10,
  2018}, pages 2533--2544, 2018.

\bibitem[GK16]{GoldnerK16}
Kira Goldner and Anna~R. Karlin.
\newblock A prior-independent revenue-maximizing auction for multiple additive
  bidders.
\newblock In {\em Web and Internet Economics - 12th International Conference,
  {WINE} 2016, Montreal, Canada, December 11-14, 2016, Proceedings}, pages
  160--173, 2016.

\bibitem[GL18]{GravinL18}
Nick Gravin and Pinyan Lu.
\newblock Separation in correlation-robust monopolist problem with budget.
\newblock In {\em Proceedings of the Twenty-Ninth Annual {ACM-SIAM} Symposium
  on Discrete Algorithms, {SODA} 2018, New Orleans, LA, USA, January 7-10,
  2018}, pages 2069--2080, 2018.

\bibitem[Gon18]{Gonczarowski18}
Yannai~A. Gonczarowski.
\newblock Bounding the menu-size of approximately optimal auctions via
  optimal-transport duality.
\newblock In {\em Proceedings of the 50th Annual {ACM} {SIGACT} Symposium on
  Theory of Computing, {STOC} 2018, Los Angeles, CA, USA, June 25-29, 2018},
  pages 123--131, 2018.

\bibitem[Har17]{hartlinebook}
Jason~D. Hartline.
\newblock {\em Approximation and Mechanism Design}.
\newblock 2017.

\bibitem[HN13]{HartN13}
Sergiu Hart and Noam Nisan.
\newblock The menu-size complexity of auctions.
\newblock In {\em the 14th ACM Conference on Electronic Commerce (EC)}, 2013.

\bibitem[HN17]{HartN17}
Sergiu Hart and Noam Nisan.
\newblock Approximate revenue maximization with multiple items.
\newblock {\em J. Economic Theory}, 172:313--347, 2017.

\bibitem[HR09]{HartlineR09}
Jason~D. Hartline and Tim Roughgarden.
\newblock Simple versus optimal mechanisms.
\newblock In {\em ACM Conference on Electronic Commerce}, pages 225--234, 2009.

\bibitem[HR15]{HartR15}
Sergiu Hart and Philip~J. Reny.
\newblock {Maximizing Revenue with Multiple Goods: Nonmonotonicity and Other
  Observations}.
\newblock {\em Theoretical Economics}, 10(3):893--922, 2015.

\bibitem[KMS{\etalchar{+}}19]{KothariMSSW19}
Pravesh Kothari, Divyarthi Mohan, Ariel Schvartzman, Sahil Singla, and
  S.~Matthew Weinberg.
\newblock Approximation schemes for a buyer with independent items via
  symmetries.
\newblock In {\em the 60th Annual IEEE Symposium on Foundations of Computer
  Science (FOCS)}, 2019.

\bibitem[KW12]{KleinbergW12}
Robert Kleinberg and S.~Matthew Weinberg.
\newblock {Matroid Prophet Inequalities}.
\newblock In {\em the 44th Annual ACM Symposium on Theory of Computing (STOC)},
  2012.

\bibitem[LP18]{LiuP18}
Siqi Liu and Christos{-}Alexandros Psomas.
\newblock On the competition complexity of dynamic mechanism design.
\newblock In {\em Proceedings of the Twenty-Ninth Annual {ACM-SIAM} Symposium
  on Discrete Algorithms, {SODA} 2018, New Orleans, LA, USA, January 7-10,
  2018}, pages 2008--2025, 2018.

\bibitem[LY13]{LiY13}
Xinye Li and Andrew Chi-Chih Yao.
\newblock On revenue maximization for selling multiple independently
  distributed items.
\newblock {\em Proceedings of the National Academy of Sciences},
  110(28):11232--11237, 2013.

\bibitem[MR16]{MorgensternR16}
Jamie Morgenstern and Tim Roughgarden.
\newblock Learning simple auctions.
\newblock In {\em Proceedings of the 29th Conference on Learning Theory, {COLT}
  2016, New York, USA, June 23-26, 2016}, pages 1298--1318, 2016.

\bibitem[MS15]{MaS15}
Will Ma and David Simchi{-}Levi.
\newblock Reaping the benefits of bundling under high production costs.
\newblock {\em CoRR}, abs/1512.02300, 2015.

\bibitem[Mye81]{Myerson81}
Roger~B. Myerson.
\newblock {Optimal Auction Design}.
\newblock {\em Mathematics of Operations Research}, 6(1):58--73, 1981.

\bibitem[Rub16]{Rubinstein16}
Aviad Rubinstein.
\newblock Settling the complexity of computing approximate two-player nash
  equillibria.
\newblock In {\em FOCS}, 2016.

\bibitem[RW15]{RubinsteinW15}
Aviad Rubinstein and S.~Matthew Weinberg.
\newblock Simple mechanisms for a subadditive buyer and applications to revenue
  monotonicity.
\newblock In {\em Proceedings of the 16th ACM Conference on Electronic
  Commerce}, 2015.

\bibitem[Syr17]{Syrgkanis17}
Vasilis Syrgkanis.
\newblock A sample complexity measure with applications to learning optimal
  auctions.
\newblock In {\em Advances in Neural Information Processing Systems 30: Annual
  Conference on Neural Information Processing Systems 2017, 4-9 December 2017,
  Long Beach, CA, {USA}}, pages 5358--5365, 2017.

\bibitem[Tha04]{Thanassoulis04}
John Thanassoulis.
\newblock Haggling over substitutes.
\newblock {\em Journal of Economic Theory}, 117:217--245, 2004.

\bibitem[VM07]{VincentM07}
D.~R. Vincent and A.~M. Manelli.
\newblock {Multidimensional Mechanism Design: Revenue Maximization and the
  Multiple-Good Monopoly}.
\newblock {\em Journal of Economic Theory}, 137(1):153--185, 2007.

\bibitem[Yao15]{Yao15}
Andrew Chi-Chih Yao.
\newblock {An n-to-1 bidder reduction for multi-item auctions and its
  applications}.
\newblock In {\em the Twenty-Sixth Annual ACM-SIAM Symposium on Discrete
  Algorithms (SODA)}, 2015.

\end{thebibliography}
\appendix
\pagebreak

\section{Summary of Known Results}\label{app:tables}

Table~\ref{tab:one-ind} and Table~\ref{tab:one-cor} presents the best results known for one additive buyer with item values sampled independently or arbitrarily, respectively.
Additionally, Table~\ref{tab:many-ind} presents the best results known for many additive buyers in the independent setting.
In each cell there is the known upper and lower bounds of the ratio between the corresponding column quantity and row quantity, and the source of the result. For example, in Table~\ref{tab:one-ind}, the table entry that corresponds to the row marked by $\max\{\SRev,\BRev\}$ and column marked by \Rev\ there is the upper bound of $5.2$ that slightly improves the 6-approximation from Theorem~\ref{thm:main} for the ratio $\Rev/\max\{\SRev,\BRev\}$ which holds for every distribution $D$. Results that are implied from other results, point to the results that imply them.

\begin{table}[h]
\caption{One buyer, independent item values. When the top number in a box is $x$, it means that $x$ times the row quantity exceeds the column quantity for all distributions. When the bottom number is $x$, it means there exists a distribution such that the row quantity times $x$ does not exceed the column quantity.}
\begin{center}
    \begin{tabular}{ || c || c | c | c ||}
    \hline \hline
      & max\{\SRev,\BRev\} & \Rev \\ \hline\hline
       &  & \\
    \SRev  & $O(\log n)$~[$\rightarrow$] & $O(\log n)$~\cite{LiY13} \\
          & $\Omega(\log n)$~\cite{HartN17} & $\Omega(\log n)$~[$\leftarrow$] \\
           &  & \\
\hline
     &   & \\
    $\max\{\SRev,\BRev\}$   & 1 & {5.2~\cite{MaS15}} \\
        & 1 &  $2$~\cite{Rubinstein16} \\
         &  & \\
\hline \hline
    \end{tabular}
\end{center}
\label{tab:one-ind}
\end{table}


\begin{table}[h]
\caption{One buyer, correlated item values. When the top value in a row is $\infty$, it simply means that $\infty$ times the row quantity exceeds the column quantity (which is trivial). When the bottom value in a row is $\infty$, it means that there exists a distribution such that for all finite $x$, $x$ times the row quantity does not exceed the column quantity.  }
\begin{center}
    \begin{tabular}{ || c || c | c | c | c ||}
    \hline \hline
      & max\{\SRev,\BRev\} & \PRev & \Rev \\ \hline\hline
       &  &  & \\
    \SRev  & $O(\log n)$~[$\rightarrow$] & {\bf $O(\log n)$~[Thm~\ref{thm:one-cor-srev-brev}]} &$\infty$\\
          & $\Omega(\log n)$~\cite{HartN17} & $\Omega(\log n)$~[$\downarrow$] & $\infty$~[$\downarrow$] \\
           & & & \\ \hline
     & &  & \\
    $\max\{\SRev,\BRev\}$ & 1 & $O(\log n)$~[$\uparrow$] &$\infty$\\
       & 1 & {\bf $\Omega(\log n)$~[Prop~\ref{prop:lb-cor}]}  &  $\infty$~[$\downarrow$] \\
       &  & & \\ \hline
      & &  & \\
    \PRev & 1 & 1 &$\infty$\\
    & 1  & 1  & $\infty$~\cite{BriestCKW15,HartN13} \\
    & &  & \\ \hline \hline
    \end{tabular}
\end{center}
\label{tab:one-cor}
\end{table}

\begin{table}[h]
\caption{Many buyers, independent item values}
\begin{center}
    \begin{tabular}{ || c || c | c | c | c ||}
    \hline \hline
      & max\{\SRev,\BRev\} & \PRev & \Rev \\ \hline\hline
       &  &  & \\
    \SRev  & $O(\log n)$~[$\rightarrow$] & $O(\log n)$~[$\rightarrow$] & {\bf $O(\log n)$~[Thm~\ref{thm:manySRev}]} \\
          & $\Omega(\log n)$~\cite{HartN17} & $\Omega(\log n)$~[$\leftarrow$] & $\Omega(\log n)$~[$\leftarrow$] \\
           & & & \\ \hline
     & &  & \\
    $\max\{\SRev,\BRev\}$ & 1 & $O(\log n)$~[$\uparrow$] & $O(\log n)$~[$\uparrow$] \\
       & 1 & {\bf $\Omega(\log n)$~[Prop~\ref{prop:lb-prev-max}]}  &  $\Omega(\log n)$~[$\downarrow$] \\
       &  & & \\ \hline
      & &  & \\
    \PRev &1  & 1 & $O(\log n)$~[$\uparrow$] \\
    &1   & 1  & {\bf $\Omega(\log n)$~[Prop~\ref{prop:lb-many-iid}]} \\
    & &  & \\ \hline \hline
    \end{tabular}
\end{center}
\label{tab:many-ind}
\end{table}

Regarding Table~\ref{tab:many-ind}, recall that the middle box ($\max\{\SRev,\BRev\}$ vs.~$\PRev$) becomes $O(1)$ and $\Omega(1)$ when $D$ has either i.i.d.~buyers (Theorem~\ref{thm:previid}) or i.i.d.~items (Theorem~\ref{thm:previiditems}).

\section{Omitted Proofs from Section~\ref{sec:prelim}}\label{app:prior}
The proofs of Lemmas~\ref{lem:HN1} and~\ref{thm:HN} require some technical lemmas from~\cite{HartN17}. We include them below with proofs for completeness. In Lemma~\ref{lem:HN0} below, $D$ and $D'$ are distributions over values for disjoint sets of items for the same buyers, and $D$ and $D'$ may be dependent. By $\Rev(D, D')$ we mean the optimal revenue obtainable by selling to buyers whose values for items are sampled from the joint distribution according to $D$ and $D'$. The Marginal Mechanism lemma below essentially states that while selling disjoint sets of items jointly may achieve significantly greater revenue than selling them separately (i.e. $\Rev(D, D')$ can be significantly larger than $\Rev(D) + \Rev(D')$), it cannot exceed the \emph{welfare} of one set plus the revenue generated by the other.

Below, and throughout this section, we use the notation $D|E$ to denote the distribution $D$ conditioned on event $E$. In particular, the distribution $D'|\vec{v} = \vec{w}$ draws a valuation $(\vec{v},\vec{v}') \leftarrow (D, D')$ conditioned on $\vec{v} = \vec{w}$ and outputs $\vec{v'}$. We also use the notation $\Rev_M(D)$ to denote the revenue that a particular mechanism $M$ guarantees on distribution $D$.

\begin{lemma}\label{lem:HN0}(``Marginal Mechanism'' \cite{HartN17,CaiH13}) $\Rev(D, D') \leq \Welfare(D) + \mathbb{E}_{\vec{w} \leftarrow D}[\Rev(D' | \vec{v}=\vec{w})]$.
\end{lemma}

\begin{proof} 
We will establish a lower bound on $\Rev(D'|\vec{v}=\vec{w})$ by constructing a truthful mechanism for selling items in the support of $D'$, based on one for those in the support of $(D, D')$. To sell items in the support of $D'$, first make ``imaginary items'' for each item in the support of $D$. Then, announce that whenever buyer $j$ receives an imaginary item $i$ in the support of $D$, she will instead receive money equal to $w_{ij}$.  Note that, due to this announcement, each buyer $j$ now has a value $w_{ij}$ for each imaginary item $i$ in the support of $D$. 

Next, take any optimal mechanism $M$ for selling items in the support of $(D,D')$ and run exactly this mechanism for buyers with values drawn from $D'|\vec{v}=\vec{w}$ (with make-believe values $\vec{w}$ for items in the support of $D$). Observe that a buyer with value $\vec{v'}^j$ has exactly the same incentives in $M$ as a buyer with values $(\vec{w}^j,\vec{v'}^j)$ (because we have explicitly given them value $\vec{w}$ for all items in the support of $D$). So this mechanism is truthful, and generates revenue $\Rev_M( (D, D')|\vec{v} = \vec{w})$, minus the money awarded for the imaginary items, on distribution $D' | \vec{v} = \vec{w}$. Observe that the total money awarded is at most $\sum_i \max_j \{w_{ij}\}$ (if the buyer with highest value for each imaginary item purchases it). So we get:

$$\Rev(D' | \vec{v} = \vec{w}) \geq \Rev_M( (D, D')|\vec{v} = \vec{w}) - \sum_i \max_j \{w_{ij}\}.$$

Now, let's take an expectation over $\vec{w}$ of both sides:

$$\mathbb{E}_{\vec{w} \leftarrow D}\left[\Rev(D' | \vec{v} = \vec{w}) \right] \geq \mathbb{E}_{\vec{w} \leftarrow D} \left[ \Rev_M (D, D') | \vec{v} = \vec{w}) - \sum_i\max_j\{ w_{ij}\}\right]$$
$$ = \Rev_M(D, D') - \Welfare(D) = \Rev(D, D') - \Welfare(D).$$
\end{proof}

\begin{prevproof}{Lemma}{lem:HN1}
This is an immediate corollary of Lemma~\ref{lem:HN0}. As $D$ and $D'$ are independent, $\Rev(D' | \vec{v}=\vec{w}) = \Rev(D')$, for all $\vec{w}$. 
\end{prevproof}

\begin{lemma}\label{lem:HN01}(``Sub-Domain Stitching'' \cite{HartN17}) Let $S_1,\ldots,S_k$ form a partition of $\mathbb{R}_+^{nm}$ and let $s_i = \Pr_{\vec{v} \leftarrow D}[\vec{v} \in S_i]$. Then $\sum_i s_i \cdot \Rev(D | \vec{v} \in S_i) \geq \Rev(D)$.
\end{lemma}

\begin{proof}
Let $M$ be the optimal mechanism for $D$, and $\Rev_M(D)$ denote the revenue of $M$ when valuations are sampled from $D$. Then we have $\Rev_M(D) = \sum_i s_i \cdot \Rev_M(D | \vec{v} \in S_i)$. We also clearly have $\Rev(D) = \Rev_M(D)$, and $\Rev_M(D|\vec{v} \in S_i) \leq \Rev(D | \vec{v} \in S_i)$ for all $i$, proving the lemma.
\end{proof}


In the lemma below, again think of $D$ and $D'$ as independent distributions for the same buyers over disjoint sets of items. 

\begin{lemma}\label{lem:HN02}(``Marginal Mechanism on Sub-Domain''~\cite{HartN17}) Let there be $m=1$ buyer, and let $S$ be any subset of $\mathbb{R}_+^{n}$, and $s = \Pr_{(\vec{v},\vec{v}')\leftarrow D \times D'}[(\vec{v},\vec{v}') \in S]$. Then $s\cdot \Rev(D\times D' | (\vec{v},\vec{v}') \in S) \leq s\cdot \Welfare(D | (\vec{v},\vec{v}') \in S) + \Rev(D')$. 
\end{lemma}

\begin{proof}
Let $T$ be the set of items in the support of $D$ and let $T'$ be the set of items in the support of $D'$.  We will use the same approach as in the proof of Lemma~\ref{lem:HN0}: design a mechanism for selling items in $T'$ given the optimal mechanism for selling items in $T\times T'$ with distribution $D\times D'\  |\ (\vec{v},\vec{v}')\in S$, {and show that this mechanism has revenue at least $s\cdot \Rev(D\times D' | (\vec{v},\vec{v}') \in S) - s\cdot \Welfare(D | (\vec{v},\vec{v}') \in S)$}. Let $M$ denote this optimal (truthful direct) mechanism. To sell items in $T'$, first draw $\vec{v} \leftarrow D$ and announce it to the buyer.  Then solicit the buyer's type $\vec{v}'$ and consider the allocation $x$ and price $p$ returned by $M$ for buyer type $(\vec{v},\vec{v}')$. If $(\vec{v},\vec{v}')\not\in S$, the buyer receives nothing and pays nothing. If $(\vec{v},\vec{v}')\in S$, we copy the allocation and price of $M$, using rebates {on the payment to simulate the buyer's expected value for the allocated items in $T$ (so long as the resulting total price is non-negative).}
More precisely, if $\sum_{i\in T}x_i v_i\leq p$, allocate the buyer each item $i \in T'$ {independently} with probability $x_i$ 
and charge the buyer $p-\sum_{i\in T}x_iv_i$.  Otherwise (if $\sum_{i\in T}x_iv_i> p$), the buyer receives nothing and pays nothing.



If the buyer has a value $\vec{v}'$ such that $(\vec{v},\vec{v}') \notin S$, we can make no guarantees about what the buyer will report. Indeed, they may wish to lie about $\vec{v}'$ because it will cause them to get items and rebates that they otherwise wouldn't. But certainly the payment made by such $\vec{v}'$ is non-negative (because we removed any options where the rebates exceed the original payment).

Similarly, if the buyer has a value $\vec{v}'$ such that $(\vec{v},\vec{v}') \in S$, but the rebates exceed the price, the buyer may misreport.  However, the payment made by these buyers is non-negative and so at least as large as the payment in $M$ minus the rebates.

Finally, if the buyer has a value $\vec{v}'$ such that $(\vec{v},\vec{v}') \in S$ and the rebates do not exceed the price then, because $M$ is truthful, we can guarantee that the buyer prefers to tell the truth. Indeed, because the buyer will get the rebates by telling the truth, their utility is exactly the same as in $M$. Their utility for any lie is at most the utility by reporting that lie to $M$, since they may additionally lose the rebates. So the buyer's payment is equal to their payment in $M$ minus the rebate. 

To summarize, we have argued that we get non-negative revenue from all $(\vec{v},\vec{v}') \notin S$, and that from all $(\vec{v},\vec{v}') \in S$ we get revenue at least their payment in $M$, minus their rebate.

Putting this together, this means that the revenue of our mechanism is at least as large as the revenue obtained by $M$, \emph{only counting revenue when} $(\vec{v},\vec{v}') \in S$ (which is $s \cdot \Rev_M(D\times D' | (\vec{v},\vec{v}') \in S) = s \cdot \Rev(D \times D'|(\vec{v},\vec{v}') \in S)$), minus the rebates given to $(\vec{v},\vec{v}') \in S$. Observe that the maximum given back in rebates to such values is $s \cdot \Welfare(D|(\vec{v},\vec{v}') \in S)$. So we have a mechanism for selling items in $T'$ with distribution $D'$ guaranteeing revenue at least $s \cdot \Rev(D \times D'|(\vec{v},\vec{v}') \in S) - s\cdot\Welfare(D|(\vec{v},\vec{v}') \in S)$, completing the proof.\footnote{Briefly observe that the above analysis considers a fixed menu, and allows the buyer to purchase whatever option they like from that menu. By the taxation principle, this is equivalent to a truthful mechanism.}
\end{proof}

\begin{prevproof}{Lemma}{thm:HN}
We first prove the lemma in the case of $m=1$.  For each item $i$, let $S_i$ be the set of types where item $i$ is the buyer's favorite item, tie-breaking lexicographically (i.e., $S_i=\{\vec{v}\ |\ \forall j\not=i, v_i > v_j\} $, where `$>$' tie-breaks lexicographically), let $D^{(i)}$ be the conditional distribution of $D$ given the event that $\vec{v}\in S_i$, let $s_i$ be the probability of this event, and let $M$ be the optimal mechanism for distribution $D$.  Then $\Rev_M(D)\leq\sum_is_i\Rev_M(D^{(i)})\leq\sum_is_i\Rev(D^{(i)})$. By Lemma~\ref{lem:HN02}, $$s_i\cdot \Rev(D^{(i)}) \leq s_i\cdot \Welfare(D^{(i)}_{-i}) + \Rev(D_i).$$ Furthermore, we claim that $s_i\cdot  \Welfare(D^{(i)}_{-i}) \leq (n-1)\cdot  \Rev(D_i)$. To see this, observe that one truthful mechanism for selling just item $i$ first samples $\vec{v}_{-i}\leftarrow D_{-i}$, and then sets a price of $\max_{j \not= i} \{v_j\}$. Conditioned on $\vec{v} \in S_i$, the item will always sell (by the definition of $S_i$), and will generate revenue $\max_{j \not= i} \{v_j\} \geq \frac{1}{n-1} \sum_{j \not= i} v_j$. So the item sells with probability $s_i$, and makes expected revenue at least $ \frac{1}{n-1}\Welfare(D^{(i)}_{-i})$ when this occurs, implying $\Rev(D_i) \geq \frac{s_i}{n-1}  \Welfare(D^{(i)}_{-i})$, as claimed. Plugging this into the bound from Lemma~\ref{lem:HN02} above, we have now shown that $s_i\cdot  \Rev(D^{(i)}) \leq n\cdot  \Rev(D_i)$. Summing over all $i$, we get the desired bound $\Rev(D)\leq n\cdot\SRev(D)$ for $m=1$.

We conclude by proving the $m > 1$ case. To extend to $m>1$ buyers, observe that any truthful $m$-buyer mechanism $M$ induces $m$ truthful single-buyer mechanisms $M_1,\ldots,M_m$ such that $\Rev_M(D) = \sum_j \Rev_{M_j}(D^j)$ (i.e., for each $M_j$, just sample $m-1$ make-believe buyers and have them play $M$). As  $\Rev_{M_j}(D^j) \leq n\cdot \SRev(D^j) \leq n \cdot \SRev(D)$ for every $j$, the mechanism $M$ cannot have revenue more than $n\cdot m\cdot \SRev(D)$.\end{prevproof}
\section{Omitted Proofs from Section~\ref{sec:core}}\label{app:core}

\begin{prevproof}{Lemma}{lem:LY1}
One could sell item $i$ using a second price auction with reserve $t_ir_i$ to guarantee revenue at least $p_i t_i r_i$. If $p_i > 1/t_i$, then this contradicts the fact that the optimal revenue is $r_i$. 
\end{prevproof}

The proof of Lemma~\ref{lem:LY2} (and future proofs) will rely on the following well-known fact about single-item auctions. Note that the proof for $m=1$ buyer is straight-forward, but the proof for $m>1$ buyer is not obvious --- we refer the reader to~\cite{DevanurHP16} (Theorem 2.1) for a proof of an even stronger claim. Notice that the lemma below does \emph{not} hold when there are $n > 1$ items~\cite{HartR15}. 

\begin{lemma}\label{lem:monotone} Let there be $n=1$ item and $m$ buyers, whose values are drawn independently. Further, consider two instances, $D$ and $D^+$ such that every marginal of $D^+$ stochastically dominates the corresponding marginal of $D$. Then $\Rev(D^+) \geq \Rev(D)$. 
\end{lemma}

\begin{prevproof}{Lemma}{lem:LY2}	
Each distribution $D_{ij}^C$ is stochastically dominated by $D_{ij}$. By Lemma~\ref{lem:monotone}, we conclude that each $\Rev(D^C_i) \leq r_i$. 

It is possible to obtain revenue $p_i \Rev(D_i^T)$ when selling to buyers from $D_i$. Simply use whatever mechanism is used to obtain revenue $\Rev(D_i^T)$. With probability $p_i$, the buyers will be sampled from $D_i^T$ and yield this much revenue. Therefore, we must have $p_i \Rev(D_i^T) \leq r_i$.
\end{prevproof}

\begin{prevproof}{Lemma}{lem:HN2}
This is a direct application of Lemma~\ref{lem:HN01}. Applied here, observe that the supports of $D_A$ form a partition of the support of $D$ when taken over all $A$.
\end{prevproof}
\section{Omitted Proofs from Section~\ref{sec:one-ind}}
\label{app:one-bidder-proofs}

\begin{prevproof}{Lemma}{lem:var}
$\Var(F) \leq E_{X \sim F}[X^2]$. As the optimal revenue of $F$ is at most $y$, we know that $Pr_{X \sim F}[X \geq x] \leq y/x$ for all $x$. {Additionally, $Pr_{X \sim F}[X \geq x] \leq 1$, as it is a probability.}  So
\begin{align*}
E_{X \sim F}[X^2] & = \int_0^{t^2y^2} Pr_{X \sim F}[X^2 \geq x] dx \\
& \leq \int_0^{y^2} dx + \int_{y^2}^{t^2y^2} (y/\sqrt{x})dx \\
& = y^2 + 2y\sqrt{x}|_{y^2}^{t^2y^2} \\
& = y^2 + 2ty^2 - 2y^2 = (2t-1)y^2.
\end{align*}
\end{prevproof}


\section{Omitted Proofs from Section~\ref{sec:multi}}
\label{app:many-bidder-proofs}
Analysis of our examples, as well as the proofs of Theorems~\ref{thm:previid} and~\ref{thm:previiditems}, use the following theorem of~\cite{ChawlaHMS10}, which describes a simple constant-factor approximation mechanism to $\Rev(D)$ in the case of a single item ($n=1$) that will be easier to analyze than $\Rev(D)$ itself. Below, a {\emph{posted-price mechanism} simply sets a price $p^j$ on the item for buyer $j$, and lets the lexicographically-first buyer $j$ whose value exceeds $p^j$ take the item and pay $p^j$ (we refer to a posted-price mechanism as \emph{anonymous} if all $p^j$ are i.i.d.).} {To properly use the result of~\cite{ChawlaHMS10} for arbitrary distributions, we will also let $p^j$ be a random variable. That is, the mechanism visits buyers one at a time in lexicographical order. When visiting buyer $j$, it draws the random variable $p^j$, and offers the item at price $p^j$ to buyer $j$.} Throughout this section, when discussing a random variable $X$, we will assume there exists an $x$ such that $\Pr[X \geq x] = c$ for any $c \in (0,1)$. If $X$ is continuous, clearly such an $x$ exists. If $X$ is discrete, we will ``make $X$ continuous'' by additionally drawing a tie-breaker uniformly from $[0,1]$ and attaching it to each draw of $X$ (and when comparing two draws from $X$ with the same value, we say the one with larger tie-breaker is larger).\footnote{So for example, if $X$ is a point-mass at $1$, then the value $1$ with tie-breaker $1/3$ is an $x$ satisfying $\Pr[X \geq x = 2/3]$.}

{We state two versions of their theorem below, which differ only in the choice of prices to set. The distinction is not necessary when all distributions $D_j$ are regular\footnote{A one-dimensional distribution is regular if $x - \frac{1-F(x)}{f(x)}$ is monotone non-decreasing} (that is, the prices in Theorem~\ref{thm:chms} can be taken to be deterministic if all $D_j$ are regular). For arbitrary distributions, however, the prices can either {be deterministic, as in Theorem \ref{thm:chms2}, or be randomized but also ensure that the probability of sale is at most half, as in Theorem \ref{thm:chms}.}
Theorem~\ref{thm:chms} will be used to prove Theorems~\ref{thm:previid} and~\ref{thm:previiditems}, and Theorem~\ref{thm:chms2} will be used to analyze our examples.}

\begin{theorem}[\cite{ChawlaHMS10}, version 1]\label{thm:chms} Let there be a single item ($n=1$) and $m$ buyers. Then there exists a {posted-price} mechanism that achieves expected revenue at least $\Rev(D)/2$. Moreover, if $p^1,\ldots, p^m$ denote the {(random variable)} prices used, then {$\sum_j \Pr[v^j \geq p^j]  \leq 1/2$}. Moreover, if the buyers are i.i.d., then {the random variables $\{p^j\}_{j \in [m]}$ are i.i.d. as well.}
\end{theorem}

\begin{theorem}[\cite{ChawlaHMS10}, version 2]\label{thm:chms2} Let there be a single item ($n=1$) and $m$ buyers. Then there exists a {posted-price} mechanism that achieves expected revenue at least $\Rev(D)/2$. Moreover, the prices $p^1,\ldots, p^m$ used are deterministic. Moreover, if the buyers are i.i.d., then $p^1,\ldots, p^m$ are identical.
\end{theorem}

{The following corollary simply observes that the revenue of a posted-price mechanism with (random) prices $p^1,\ldots, p^m$ is clearly upper-bounded by {$\sum_j \mathbb{E}_{p \leftarrow p^j}\left[p \cdot \Pr[v^j \geq p]\right]$.}}

\begin{corollary}\label{cor:chms}
	Let there be a single item ($n=1$) and $m$ buyers. Then there exist {random variables} $p^1,\ldots, p^m$ such that {$\sum_j \mathbb{E}_{p \leftarrow p^j}\left[p \cdot \Pr[v^j \geq p]\right] \geq \Rev(D)/2$}, and $\sum_j \Pr[v^j \geq p^j]  \leq 1/2$. Moreover, if the buyers are i.i.d., then {the random variables $\{p^j\}_{j \in [m]}$ are i.i.d. as well.}
\end{corollary}




\subsection{Proof of Theorems~\ref{thm:previid} and~\ref{thm:previiditems}}\label{app:iidproofs}
Recall that our proof outline first defines a set $S$ as \emph{separable} or \emph{bundlable}, based on whether we will target it with $\SRev(D)$ or $\BRev(D)$. We first begin by providing these definitions. 

\begin{definition}[Separable set] Say that a set $S$ is $\alpha$-separable for $D$ if $\SRev(D_S) \geq \alpha \cdot \BRev(D_S)$. 
\end{definition}

\begin{definition}[Bundlable set] Say that a set $S$ is $\beta$-bundlable for buyer $j$ and distribution $D$ if there exists a price $p$ such that $\Pr[\sum_{i \in S} v_{ij} \geq p] \geq 1/2$ and $p \geq \beta\cdot \BRev(D_S)$.
\end{definition}

We now show that $\SRev(D)$ covers revenue from separable sets, and $\BRev(D)$ covers revenue from bundlable sets. Importantly, the latter claim requires that it is the same buyer $j$ which witnesses that $S$ is $\beta$-bundlable \emph{for all $S$ that are not separable}. This aspect is what enables Theorem~\ref{thm:previid} when buyers are i.i.d. (because if the condition holds for one buyer, it holds for all of them), or when items are i.i.d. (as we will argue that there is a ``dominant buyer'' to whom we can restrict attention for all $S$) but not for arbitrary instances (see Proposition~\ref{prop:lb-prev-max} for a counterexample). 

\begin{lemma}\label{lem:separable} Let $\mathcal{S}$ be any collection of disjoint subsets of items such that $S$ is $\alpha$-separable for $D$ for all $S \in \mathcal{S}$. Then $\SRev(D) \geq \alpha \cdot \sum_{S \in \mathcal{S}} \BRev(D_S)$.
\end{lemma}
\begin{proof}
As $\SRev(D)$ sells all items, {$\mathcal{S}$ is a collection of disjoint subsets of items, and items have non-negative values,
}  we clearly have $\SRev(D) \geq \sum_{S \in \mathcal{S}} \SRev(D_S)$. By definition of $\alpha$-separability, we have that $\SRev(D_S) \geq \alpha \cdot \BRev(D_S)$ for all $S \in \mathcal{S}$, completing the inequality.
\end{proof}

\begin{lemma}\label{lem:bundlable} Let $\mathcal{B}$ be any collection of disjoint subsets of items such that there exists a buyer $j$ such that $S$ is $\beta$-bundlable for $j$ and $D$ for all $S \in \mathcal{B}$. Then $\BRev(D) \geq \frac{\beta}{4} \cdot \sum_{S \in \mathcal{B}} \BRev(D_S)$. 

\end{lemma}
\begin{proof}

Let $p_S$ denote the price promised by $\beta$-bundlability for set $S\in \mathcal{B}$.  Consider setting price $p := \sum_{S \in \mathcal{B}} p_S/2$ on the grand bundle of all items and running a posted-price mechanism with price $p$ (the same $p$ for all buyers). We claim that buyer $j$ will choose to purchase the grand bundle with probability at least $1/2$, and therefore $\BRev(D) \geq \sum_{S \in \mathcal{B}}p_S/4$. As we are promised that $p_S \geq \beta \cdot \BRev(D_S)$ for all $S \in \mathcal{B}$, this would imply that $\BRev(D) \geq \frac{\beta}{4} \cdot \sum_{S \in \mathcal{B}} \BRev(D_S)$.

It remains to establish that buyer $j$ will choose to purchase the bundle with probability at least $1/2$ at price $p$. Consider a random variable $V_S$ which is equal to $p_S$ whenever $\sum_{i\in S} v_{ij} \geq p_S$ and $0$ otherwise. Observe that $\sum_{i \in S} v_{ij}$ stochastically dominates $V_S$ (immediately from the definition of $V_S$: we defined it to first sample $\sum_{i \in S} v_{ij}$ and then lower it to either $p_S$ or $0$). Consider finally the random variable $W_S$ which is equal to $p_S$ with probability $1/2$ and $0$ otherwise. It is also clear that $V_S$ stochastically dominates $W_S$ for all $S\in \mathcal{B}$ (because $\beta$-bundability of $S$ guarantees that $V_S = p_S$ with probability at least $1/2$). Therefore $\sum_{i \in S} v_{ij}$ stochastically dominates $W_S$ as well. We proceed to analyze the random variable $W = \sum_{S \in \mathcal{B}} W_S$.

Observe that $W$ is \emph{symmetric about its mean}, which is $p$. That is, for all $x$ , $\Pr[W = p+x] = \Pr[W =p- x]$. To see this, couple draws $\langle W_S \rangle_{S \in \mathcal{B}}$ with $\langle p_S - W_S\rangle_{S \in \mathcal{B}}$. Observe that because $\Pr[W_S = 0] = 1/2 = \Pr[W_S = p_S]$ for all $S\in \mathcal{B}$, the two coupled draws are equally likely. Moreover, observe that for all coupled draws, $\sum_{S \in \mathcal{B}} W_S + \sum_{S \in \mathcal{B}} (p_S - W_S) = \sum_{S \in \mathcal{B}} p_S = 2p$. Therefore, all coupled draws have equal distance from $p$. As this coupling maps between outcomes with equal probability, and these outcomes are symmetric about $p$, the entire random variable $W$ is symmetric about $p$. We therefore conclude that $\Pr[W \geq p] = 1/2$, and as $\sum_i v_{ij} $ stochastically dominates $W$, we therefore get that $\Pr[\sum_i v_{ij} \geq p] \geq 1/2$. Therefore, selling the entire grand bundle at price $p$ results in revenue at least $p/2$ (even if only sold to buyer $j$).
\end{proof}

{Next, we establish that every set $S$ must be either separable, or bundlable for some buyer $j$ (Proposition \ref{prop:seporbund}). We will later use this to prove Theorem \ref{thm:previid} and Theorem \ref{thm:previiditems}, showing that  \emph{when either buyers or items are i.i.d.} there is an approximately-optimal partition mechanism and buyer $j$ such that \emph{every set in the partition that is not separable, is bundlable for that buyer $j$}.}

The remainder of this section is dedicated to proving this claim {and the theorems it implies}, and we begin with some setup and technical lemmas. First, recall by Theorem~\ref{thm:chms} that {we may relate $\BRev(D_S)$ to the revenue of a posted-price mechanism (because $\BRev(D_S)$ is just the revenue of the optimal single-item auction for the ``item'' $S$). We define the following notation below (which is used for the rest of this section):}

\begin{itemize}
\item Denote $p_S:= \BRev(D_S)/2$, we introduce this redundant notation for convenience.
\item $p^1_S,\ldots, p^m_S$ {denote}
the {random variables} guaranteed to exist by Theorem~\ref{thm:chms}, for distribution $D_S$.
\item {Denote} $q^j_S:= \Pr[\sum_{i \in S} v_{ij} \geq p^j_S]$. {Recall that both $\sum_{i\in S} v_{ij}$ and $p^j_S$ are random variables.}
\end{itemize}

Recall that {by Corollary \ref{cor:chms} }, $\BRev(D_S) \leq 2\cdot {\sum_{j \in [m]}\mathbb{E}_{p \leftarrow p_S^j}\left[p \cdot \Pr[\sum_{i \in S} v_{ij} \geq p]\right]}$, and $\sum_j q^j_S \leq 1/2$.

We will now fix attention to a single set $S$, and wish to understand, for a single buyer {$j$}, whether the events in which $\sum_{i \in S} v_{ij}$ is large are driven mostly by a few \emph{large}-{value} items, or several items with \emph{tiny} value (below $L$ stands for large, and $T$ for tiny).

\begin{itemize}
\item For each item $i$, let $t_i$ be such that $\Pr[v_i^* \geq t_i] = 1/2$ (i.e. $t_i$ is the median of the random variable $v_i^*$. Recall that $v_i^* := \max_{j \in [m]} \{v_{ij}\}$). 
\item Define $L_{ij}:={v_{ij} \cdot \mathbb{I}(v_{ij} \geq t_i)}$.
\item Define $T_{ij}:=v_{ij} \cdot \mathbb{I}(v_{ij} < t_i)$.
\item {Define $C_{ijp}:= \min\{L_{ij},p\}$ (think of this as $L_{ij}$ ``capped at $p$'').}
\item Let $R_i^j$ denote the optimal achievable revenue selling a single item to a single buyer with value distributed according to $L_{ij}$ (observe that $R^j_i=\max_{x \geq t_i}\{x \cdot \Pr[v_{ij} \geq x]\}$). For ease of notation, also define $R^j_S:=\sum_{i \in S} R_i^j$.
\end{itemize}


Next, we will conclude a few basic properties of the above-defined random variables. Importantly, we will start connecting the expectation and variance of these random variables to $\SRev(D_S)$.

\begin{lemma}\label{lem:srevS}$\SRev(D_S) \geq \sum_{i \in S} t_i/2$. Also, $\SRev(D_S) \geq \sum_{j\in [m]} R_S^j/2$.
\end{lemma}
\begin{proof}
To see the first bound, simply set a posted price of $t_i$ on item $i$, for all $i$. By definition of $t_i$, item $i$ will be purchased with probability exactly $1/2$, giving expected revenue $t_i/2$. Summing over all items gives the bound. 

To see the second bound,  use a posted-price mechanism which sets a price for buyer $j$ to purchase item $i$ of ${z_{ij}=}\arg\max_{x \geq t_i} \{x \cdot \Pr[v_{ij}\geq x]\}$. Then visit the buyers in arbitrary order and allow them to purchase any remaining items, offering an available item $i$ to buyer $j$ at its personalized price $z_{ij}$. Because for any buyer $j$ the price $z_{ij}$ of item $i$ is always at least $t_i$, this means that for any buyer $j$, item $i$ is available for buyer $j$ to purchase with probability at least $1/2$. {When item $i$ is available for buyer $j$, we get expected revenue exactly $R^j_i$ from selling item $i$ to buyer $j$. Therefore, the total revenue achieved from this scheme (which sells items separately) is at least $\sum_{i \in S} \sum_{j \in [m]} R^j_i/2 = \sum_{j \in [m]} R^j_S/2$.}
\end{proof}

\begin{corollary}\label{cor:tiny} For any buyer $j$, $\sum_{i \in S} T_{ij} \leq 2\cdot \SRev(D_S)$ with probability $1$. 
\end{corollary}
\begin{proof}
Each $T_{ij}$ is a random variable supported on $[0,t_i]$, so their sum is at most $\sum_{i \in S} t_i$, with probability $1$. Also, $\sum_{i \in S} t_i \leq 2\cdot \SRev(D_S)$, by Lemma~\ref{lem:srevS}.
\end{proof}

\begin{corollary}\label{cor:large} For any buyer $j$, $\Var(\sum_{i \in S} {C_{ijp}}) \leq 2\cdot {p} \cdot R^j_S$. 
\end{corollary}
\begin{proof}
Each $C_{ijp}$ is a random variable supported on $[0,{p}]$, whose optimal revenue is at most $R_i^j$ {(because $C_{ijp}$ is stochastically dominated by $L_{ij}$, whose optimal revenue is $R_i^j$)}. Therefore, by Lemma~\ref{lem:var}, $\Var({C_{ijp}}) \leq 2{p} R_i^j$. As all ${C_{ijp}}$ are independent, we have that $\Var(\sum_{i \in S} {C_{ijp}}) = \sum_{i \in S} \Var({C_{ijp}}) \leq 2{p}\sum_{i \in S} R_i^j = 2{p}R^j_S$.
\end{proof}

We now quickly apply Chebyshev's inequality using Corollary~\ref{cor:large}.

\begin{corollary}\label{cor:large2} {For any buyer $j$, set $S$, and price $p$: if\ \ $\mathbb{E}[\sum_{i \in S} C_{ijp}] \leq p/4$, then $\Pr[\sum_{i \in S} C_{ij} \geq p/2] \leq \frac{32R^j_S}{p}$. If instead $\mathbb{E}[\sum_{i \in S} C_{ijp}] \geq p/4$, then $\Pr[\sum_{i \in S} C_{ijp} \geq p/8] \geq 1-\frac{128 R^j_S}{p}$.}
\end{corollary}
\begin{proof} We apply Chebyshev's inequality to the random variable $\sum_{i \in S} {C_{ijp}}$. Corollary~\ref{cor:large} establishes that $\Var(\sum_{i \in S} {C_{ijp}}) \leq 2{p}R^j_S$, and therefore

$$\Pr\left[\left|\sum_{i \in S} {C_{ijp}} - \mathbb{E}\left[\sum_{i \in S} {C_{ijp}}\right]\right| > x\right] \leq \frac{2{p}R^j_S}{x^2}.$$

For the first case of the corollary statement, the deviation must be at least $x = {p}/4$, resulting in a bound of {$\frac{32R^j_S}{p}$}. For the second case to \emph{not} occur, the deviation must be at least $x = {p}/8$, resulting in a (lower) bound of {$1-\frac{128R^j_S}{p}$}.
\end{proof}

{Now, we state our two main technical propositions. We first show that a ``good'' (Definition~\ref{def:good} immediately below) buyer always exists when $S$ is not $(1/1024)$-separable. We later establish that $S$ is either separable, or bundlable for any good buyer. To get intuition how we will eventually leverage these when items or buyers are i.i.d., observe that when buyers are i.i.d., either all buyers are good, or none are good (and therefore by Proposition~\ref{prop:good}, they are all good whenever $S$ is not separable). On the other hand, when items are i.i.d., observe that whether a buyer is good depends only on $|S|$ (we will later use this to show that the \emph{same buyer} is good for every non-separable $S$ in the partition).}

\begin{definition}\label{def:good}A buyer $j$ is \emph{good for $S$} if there exists a $p \geq p_S$ such that $\mathbb{E}[\sum_{i \in S} C_{ijp}] \geq p/4$. 
\end{definition}

\begin{proposition}\label{prop:good} If $S$ is not $(1/1024)$-separable, there exists at least one buyer $j$ which is good for $S$.
\end{proposition}
\begin{proof}
We will make use of the fact that $\sum_{j \in \left[m\right]} \mathbb{E}_{p \leftarrow p^j_S}\left[p \cdot \Pr\left[\sum_{i \in S} v_{ij}\geq p\right]\right] \geq \BRev(D_S)/2 = p_S$ (by definition, as promised by Theorem~\ref{thm:chms}). We will derive a contradiction by showing that this sum is too small if no $j$ is good for $S$, and $S$ is not $(1/1024)$-separable.

We will now break the sum into two pieces, by expanding $\sum_{j \in \left[m\right]} \mathbb{E}_{p \leftarrow p^j_S}\left[p \cdot \Pr\left[\sum_{i \in S} v_{ij} \geq p\right]\right]$. By simply splitting into contributions from cases where $p$ is {smaller than $p_S$ and when it is larger,} we can write:

\begin{align*}
\sum_{j \in\left[m\right]}\mathbb{E}_{p \leftarrow p^j_S}\left[p \cdot \Pr\left[\sum_{i \in S} v_{ij} \geq p\right]\right] &= \sum_{j \in \left[m\right]}\mathbb{E}_{p \leftarrow p^j_S}\left[p \cdot \Pr\left[\sum_{i \in S} v_{ij} \geq p\right]\cdot \mathbb{I}(p \leq p_S)\right]\\
&\quad +\sum_{j \in \left[m\right]}\mathbb{E}_{p \leftarrow p^j_S}\left[p \cdot \Pr\left[\sum_{i \in S} v_{ij} \geq p\right] \cdot \mathbb{I}(p > p_S)\right]\\
\end{align*}

We now proceed to analyze these terms separately.

\begin{lemma}\label{lem:negS}$\sum_{j \in \left[m\right]}\mathbb{E}_{p \leftarrow p^j_S}\left[p \cdot \Pr\left[\sum_{i \in S} v_{ij} \geq p\right]\cdot \mathbb{I}(p \leq p_S)\right]\leq p_S/2$. 
\end{lemma}
\begin{proof}
The proof follows by simply expanding the sum:

\begin{align*}
\sum_{j \in \left[m\right]}\mathbb{E}_{p \leftarrow p^j_S}\left[p \cdot \Pr\left[\sum_{i \in S} v_{ij} \geq p\right]\cdot \mathbb{I}(p \leq p_S)\right]&\leq\sum_{j \in \left[m\right]} \mathbb{E}_{p \leftarrow p^j_S}\left[p_S \cdot \Pr\left[\sum_{i \in S} v_{ij} \geq p\right]\cdot \mathbb{I}(p \leq p_S)\right]\\
&\leq \sum_{j \in\left[m\right]}\mathbb{E}_{p \leftarrow p^j_S}\left[p_S \cdot \Pr\left[\sum_{i \in S} v_{ij} \geq p\right]\right]\\
&= \sum_{j \in \left[m\right]}p_S \cdot q_S^j\\
&\leq p_S/2.
\end{align*}

Indeed, the first inequality simply observes that whenever the indicator is non-zero, $p \leq p_S$. The second inequality upper bounds an indicator by $1$. The equality observes that $p_S$ is a constant, and that $q^j_S:=\Pr\left[\sum_{i \in S} v_{ij} \geq p^j_S\right]= \mathbb{E}_{p \leftarrow p^j_S}\left[\Pr\left[\sum_{i \in S} v_{ij} \geq p\right]\right]$. The final inequality makes use of the promise from Theorem~\ref{thm:chms} that $\sum_{j \in \left[m\right]} q^j_S \leq 1/2$. 
\end{proof}

Next, we transition to the second term in the sum. We will argue that if no $j$ are good for $S$, then the second term is small as well.

\begin{lemma}\label{lem:smallS} Assume that $S$ is not $(1/1024)$-separable, and assume further that {no buyer $j$ is good} for $S$. Then $\sum_{j \in \left[m\right]}\mathbb{E}_{p \leftarrow p^j_S}\left[p \cdot \Pr\left[\sum_{i \in S} v_{ij} \geq p\right] \cdot \mathbb{I}(p > p_S)\right] \leq \BRev(D_S)/16$.
\end{lemma}

\begin{proof}

We begin by observing that $v_{ij} = T_{ij} + L_{ij}$ (with probability $1$). Therefore, if we are to possibly have $\sum_{i \in S} v_{ij} \geq p$, we must also have either $\sum_{i \in S} T_{ij} \geq p/2$, or $\sum_{i \in S} L_{ij} \geq p/2$. By Corollary~\ref{cor:tiny}, however, we know that $\sum_{i \in S} T_{ij} \leq 2\cdot \SRev(D_S)$ with probability $1$. Therefore, we conclude that for any $p > 4\cdot \SRev(D_S)$, $\sum_{i \in S} v_{ij} \geq p \Rightarrow \sum_{i \in S} L_{ij} \geq p/2$. Observe even further that whenever $\sum_{i \in S} L_{ij} \geq p/2$, we must further have $\sum_{i \in S} C_{ijp} \geq p/2$ (because the only reason we would have $C_{ijp} \neq L_{ij}$ for any $i,j$ is if both are already 
at least $p$). 

In particular, this allows us to conclude that, \emph{whenever $p > 4\SRev(D_S)$}:\begin{equation}\label{eq:useful}\Pr\left[\sum_{i \in S} v_{ij} \geq p\right] \leq \Pr\left[\sum_{i \in S} C_{ijp} \geq p/2\right].\end{equation}

We can now use this to complete the following chain of inequalities:
\begin{align*}
\sum_{j \in \left[m\right]}\mathbb{E}_{p \leftarrow p^j_S}\left[p \cdot \Pr\left[\sum_{i \in S} v_{ij} \geq p\right] \cdot \mathbb{I}(p > p_S)\right] &\leq \sum_{j \in \left[m\right]} \mathbb{E}_{p \leftarrow p^j_S}\left[p \cdot \Pr\left[\sum_{i \in S} C_{ijp} \geq p/2\right] \cdot \mathbb{I}(p > p_S)\right]\\
&\leq \sum_{j \in \left[m\right]} \mathbb{E}_{p \leftarrow p^j_S}\left[p \cdot \frac{32R^j_S}{p}\cdot \mathbb{I}(p> p_S)\right]\\
&\leq \sum_{j \in \left[m\right]} 32 R^j_S\\
& \leq 64\cdot \SRev(D_S)\\
&\leq \BRev(D_S)/16.
\end{align*}

Indeed, the first line follows by Equation~\eqref{eq:useful} and the fact that $S$ is not $(1/1024)$-separable: because $S$ is not $(1/1024)$-separable, whenever the indicator is non-zero we have $p > p_S = \BRev(D_S)/2 {> 512 \cdot \SRev(D_S) >4 \cdot \SRev(D_S)}$, so we can use Equation~\eqref{eq:useful}. The second line follows directly from Corollary~\ref{cor:large2}, because of our hypothesis that no $j$ is good for $S$. The third follows by upper-bounding  an indicator random variable by $1$. The fourth follows directly from Lemma~\ref{lem:srevS}. The final inequality follows as $S$ is not $(1/1024)$-separable.
\end{proof}

Using Lemmas~\ref{lem:negS} and~\ref{lem:smallS} we can now complete the proof of Proposition~\ref{prop:good}. Assume for contradiction that $S$ is not $(1/1024)$-separable and also that {no buyer $j$ is good} for $S$. Then Lemmas~\ref{lem:negS} and~\ref{lem:smallS} immediately imply that:
\begin{align*}
\sum_{j \in \left[m\right]}\mathbb{E}_{p \leftarrow p^j_S}\left[p \cdot \Pr\left[\sum_{i \in S} v_{ij} \geq p\right]\cdot \mathbb{I}(p\leq p_S)\right]& +\sum_{j \in \left[m\right]}\mathbb{E}_{p \leftarrow p^j_S}\left[p \cdot \Pr\left[\sum_{i \in S} v_{ij} \geq p\right] \cdot \mathbb{I}(p > p_S)\right]\\
 &\leq \BRev(D_S)/4 + \BRev(D_S)/16\\
&< \BRev(D_S)/2.
\end{align*}

But Theorem~\ref{thm:chms} guarantees that $\sum_{j \in \left[m\right]}\mathbb{E}_{p \leftarrow p^j_S}\left[p \cdot \Pr\left[\sum_{i \in S} v_{ij} \geq p\right]\right] \geq \BRev(D_S)/2$, a contradiction.

The only assumptions necessary to reach a contradiction are that $S$ is not $(1/1024)$-separable, and that {no buyer $j$ is good} for $S$. So one of these assumptions must fail (and in particular, if $S$ is not $(1/1024)$-separable, it must be the latter one). 
\end{proof}

We now prove our second key technical proposition, which states that $S$ is either separable, or bundlable for any good $j$.

\begin{proposition}[Separable or Bundlable]\label{prop:seporbund} {Let buyer $j$ be} good for $S$. Then $S$ is either $(1/1024)$-separable for $D$, or $(1/16)$-bundlable for $j$ and $D$.

\end{proposition}
\begin{proof}
We assume throughout the proof that $S$ is not $(1/1024)$-separable (clearly, if this assumption is false, then $S$ is $(1/1024)$-separable and we are done). Let $p$ be {any} witness that $j$ is good for $S$. The proof will follow from the chain of inequalities below:
\begin{align*}
\Pr\left[\sum_{i \in S} v_{ij} \geq p/8\right] &\geq \Pr\left[\sum_{i \in S} C_{ijp} \geq p/8\right]\\
&\geq 1-\frac{128R^j_S}{p} \\
& \geq 1 - \frac{256 \cdot \SRev(D_S)}{\BRev(D_S)/2}\\
&\geq 1/2.
\end{align*}
The first inequality follows simply as $v_{ij}$ stochastically dominates $C_{ijp}$ for all $i,j,p$. The second follows from Corollary~\ref{cor:large2}, and the hypothesis that $j$ is good for $S$ (witnessed by $p$). The third follows as $R^j_S \leq \sum_{\ell\in [m]} R^\ell_S \leq 2\cdot \SRev(D_S)$ by Lemma~\ref{lem:srevS}, and $p \geq p_S = \BRev(D_S)/2$ by definition. The final inequality follows as $S$ is not $(1/1024)$-separable. Therefore, we have shown that buyer $j$ will purchase $S$ at price $p/8$ with probability at least $1/2$. As $p \geq  p_S = \BRev(D_S)/2$ it follows that $p/8 \geq \BRev(D_S)/16$, and thus $S$ is $(1/16)$-bundlable for $j, D$.
\end{proof}

Propositions~\ref{prop:good} and~\ref{prop:seporbund} are the key ingredients in Theorems~\ref{thm:previid} and~\ref{thm:previiditems}, which we now prove below.

	\begin{proof}[Proof of Theorem~\ref{thm:previid}] The theorem now follows by Propositions~\ref{prop:good} and~\ref{prop:seporbund} and Lemmas~\ref{lem:separable} and~\ref{lem:bundlable}. For the optimal partition of the items, let $\mathcal{S}$ denote the $(1/1024)$-separable sets, and $\mathcal{B}$ denote the $(1/16)$-bundlable sets for buyer $1$. {Because buyers are i.i.d., for all $S$ it is the case that {either every buyer is good for $S$, or no buyer is good for $S$.} Proposition~\ref{prop:good} proves that for all $S \notin \mathcal{S}$, all buyers are good for $S$. Propositions~\ref{prop:seporbund} then guarantees that all $S \notin \mathcal{S}$ are $(1/16)$-bundlable for all buyers (and in particular, buyer $1$). Therefore, $\mathcal{S}$ and $\mathcal{B}$ form a cover of the partition.} Lemma~\ref{lem:separable} guarantees that $\SRev(D) \geq \frac{1}{1024}\sum_{S \in \mathcal{S}} \BRev(D_S)$, and Lemma~\ref{lem:bundlable} guarantees that $\BRev(D) \geq \frac{1}{64} \sum_{S \in \mathcal{B}} \BRev(D_S)$. Therefore, either $\SRev(D)$ or $\BRev(D)$ guarantees a $(1/2048)$-approximation to $\sum_S \BRev(D_S) = \PRev(D)$. 
\end{proof}

\begin{proof}[Proof of Theorem~\ref{thm:previiditems}]
The theorem again follows from Propositions~\ref{prop:good} and~\ref{prop:seporbund}, and Lemmas~\ref{lem:separable} and~\ref{lem:bundlable}, although one extra step is needed. Consider the optimal partition of $[n]$ into $S_1 \sqcup\ldots \sqcup S_k$ and consider $\sum_i \BRev(D_{S_i})$. First, we observe that because the items are i.i.d., $\BRev(D_{S_i})$ depends only on $|S_i|$, and not the precise items in $S_i$ (exactly because the items are i.i.d.). Therefore, we will abuse notation and write $\BRev(x)$ instead of $\BRev(D_S)$ for any $S$ with $|S| = x$. Next, let $x$ be any value maximizing $\BRev(x)/x$ among all $x \in \{1,\ldots, n\}$. That is, bundling $x$ items together generates the most ``bang per buck'' (revenue per item bundled).  Then we get

$$\sum_i \BRev(D_{S_i}) \leq \sum_i |S_i| \cdot \BRev(x)/x \leq n \BRev(x)/x.$$

Therefore, if we instead partition the items into $\lfloor n/x\rfloor$ sets of size $x$, and one remaining set, we would get revenue at least:

$$\lfloor n/x \rfloor \cdot x \cdot \BRev(x)/x \geq (n/2) \cdot \BRev(x) /x \geq \sum_i \BRev(D_{S_i})/2.$$

The first inequality follows because $\lfloor n/x \rfloor \cdot x \geq n/2$ when $x \leq n$ (i.e., if $\geq n/2$ items are left over, it is because $x \leq n/2$, in which case another bundle can be made). Therefore, up to a factor of $2$, we may analyze $\lfloor n/x \rfloor \cdot x \cdot \BRev(x)$ instead. Importantly, observe that because items are i.i.d., and we are only considering bundles $S$ of the same size, each instance $S$ considered is identically distributed, {and therefore every buyer is either \emph{good for {every such set} $S$, or good for no $S$}. This observation, together with Proposition~\ref{prop:good}, implies that if no buyer is good for all $S$, then all $S$ are $(1/1024)$-separable. Otherwise, there is a buyer $j$ which is good for all $S$, and Proposition~\ref{prop:seporbund} states that either \emph{all $S$ are $(1/1024)$-separable} or \emph{all $S$ are $(1/16)$-bundlable} for $j$ (because all $S$ are identically-sized).} Whichever the case is, $\max\{\SRev(D),\BRev(D)\}$ guarantees at least a $(1/1024)$-approximation by Lemmas~\ref{lem:separable} and~\ref{lem:bundlable}, for a total of a $(1/2048)$-approximation after losing the factor of $2$ by restricting to the case of identically-sized $S$.

\end{proof}

\subsection{Analysis of Multi-buyer Lower Bounds}
The following lemma slightly improves concentration bounds for sums of i.i.d. Equal Revenue random variables provided in~\cite{HartN17}. 
\begin{lemma}\label{lem:er}
Let $X$ be the sum of $\ell$ i.i.d. random variables all drawn according to $\er_k$. Then:
\begin{itemize}
\item $\mathbb{E}[X] = \ell + \ell \ln(k)$.
\item $\Pr[|X - \ell - \ell \ln(k)|/(\ell + \ell \ln(k)) > \delta] \leq 2e^{-\delta^2\ell \ln(k)/(3k)}$. 
\end{itemize}
\end{lemma}
\begin{proof}
The expected value of a single draw from $\er_k$ is $\int_0^\infty 1-F(x)dx = \int_0^1 1dx + \int_1^k (1/x)dx + \int_k^\infty 0 = 1+\ln(k)$. So the expected value of $X$ is just $\ell$ times this.

To derive the second bullet, we simply apply the multiplicative Chernoff bound (see Appendix~\ref{sec:concentration} for statement). Observe that each individual draw from $\er_k$ is bounded in $[0,k]$, so after normalizing to lie in $[0,1]$, we get that the probability that the sum deviates by its expectation by more than a multiplicative factor of $\delta$ is at most $2e^{-\delta^2 (\ell +\ell \ln(k))/(3k)}$ (in the lemma statement we relax this by lower bounding $\ell \geq 0$ for cleanliness). 
\end{proof}


\begin{prevproof}{Proposition}{prop:lb-many-iid}
To prove the claim we show that $\Rev(D)\in \Omega(n\log n)$ while $\PRev(D)\in O(n)$ (actually, since $\SRev\in \Omega(n)$ it holds that $\PRev(D)\in \Theta(n)$).

To see that $\Rev(D)\in \Omega(n\log n)$ consider the mechanism that sequentially visits the $\sqrt{n}$ buyers,
allowing each to pick any set of $\sqrt{n}/2$ items that are still available, and pay $c(\sqrt{n}\log n)$ for some $c>0$ to be determined later.
For each of the first $\sqrt{n}/4$ buyers, at least $7n/8$ items are remaining, and with probability $1-e^{-\Omega(\sqrt{n})}$, the buyer has non-zero value for at least $\sqrt{n}/2$ of these items. This follows directly from the multiplicative Chernoff bound (Appendix~\ref{sec:concentration} for statement), as the expected number of items for which the buyer has non-zero value is at least $7\sqrt{n}/8$, and the buyer has non-zero value for each item independently. Conditioned on having non-zero value for {each of} at least $\sqrt{n}/2$ remaining items, pick an arbitrary set of remaining items of size $\sqrt{n}/2$ for which the buyer has non-zero value and call it $S$. Then the buyer's value for each item is $S$ is drawn from $\er_{n^{1/8}}$, and therefore by Lemma~\ref{lem:er} her expected value for this set is at least $\sqrt{n}\ln(n)/16$ (plugging in $\ell = \sqrt{n}/2$ items and $k = n^{1/8}$ and dropping the $+\ell$ term). Moreover, the probability that her value for the set falls below $\sqrt{n}\ln(n)/32$ is at most $e^{-\Omega(n^{3/8})}$ (again, just plugging into Lemma~\ref{lem:er}). Therefore, if we set $c = 1/32$, then no matter what happens for the first $j-1$ buyers, buyer $j\leq \sqrt{n}/4$ will choose to purchase with probability at least $1/2$ (in fact, much closer to $1$ than this), just over the randomness in drawing $v^j$ from $D^j$. This means that our expected revenue from the first $\sqrt{n}/4$ buyers is at least $\sqrt{n}/4 \cdot 1/2 \cdot c \sqrt{n} \ln(n) = n \ln (n) / 256$. So $\Rev(D) \in \Omega(n\log n)$. 

To see that $\PRev(D)\in O(n)$ consider any partitioning mechanism into sets $S_1,\ldots, S_x$, and denote by $s_k = |S_k|$. Denote by $\BRev(k)$ and $\SRev(k)$ the revenue obtained by bundling together and selling separately the items in $S_k$, respectively. Note that $\sum_k s_k = n$ and also that $\PRev(D) = \sum_k \BRev(k)$. 

We first analyze all $k$ such that $s_k \leq n^{1/4}$. In this case, observe that the probability that a fixed buyer has non-zero value for at least $12$ items can be upper bounded by taking a union bound over all $ \binom{s_k}{12}\leq n^{(1/4)\cdot 12} = n^3$ subsets of $S_k$ of size $12$ of the probability that the buyer has non-zero value for this entire subset $S_k$ (which is just $1/\sqrt{n}^{12} = 1/n^6$ as each event is independent). Taking another union bound over all $\sqrt{n}$ buyers yields that the probability that \emph{any} buyer has non-zero value for more than $12$ items is at most $n^{3-6+1/2} = n^{-5/2}$. So even if our mechanism achieved the maximum possible welfare from $S_k$ ($s_k \cdot n^{1/8} \leq n^{3/8}$) whenever this occurred, the total contribution to the revenue would still be $o(1)$. To finish analyzing $k$ such that $s_k \leq n^{1/4}$, it remains to analyze the case when every buyer has non-zero value for $12$ or fewer items in part $k$.

When every buyer has non-zero value for $12$ or fewer items in part $k$, we show that $24 \SRev(k) \geq \BRev(k)$. To see this, let $p$ be the anonymous posted-price promised by Theorem~\ref{thm:chms2} such that $p \cdot \Pr[\text{exists a buyer who values the bundle above $p$}] \geq \BRev(k)/2$.  Consider instead the anonymous price mechanism that sets price $p/12$ for each item separately. Then clearly, whenever some buyer is willing to pay $p$ for the grand bundle, and only $12$ of these items provide non-zero value, she is also willing to pay $p/12$ for at least one item separately. So we get that $24 \SRev(k) \geq \BRev(k)$. Finally, a trivial upper bound on $\SRev(k)$ is the revenue obtainable by selling goods separately without any supply constraint (i.e. the seller has infinitely many copies of each good instead of just one). For a single item and single buyer, the optimal price is $1$ (or any value between $1$ and $n^{1/8}$), and generates revenue $1/\sqrt{n}$. So the total revenue for selling a single item separately to $\sqrt{n}$ buyers is at most $1$, and therefore $\SRev(k) \leq s_k$. So we have now shown that $\BRev(k) \leq 24s_k + o(1)$ whenever $s_k \leq n^{1/4}$. It remains now to consider parts where $s_k > n^{1/4}$.

Next we analyze the case where $s_k > n^{1/4}$. In this case, we will argue that with high probability, no buyer has large value for the bundle $S_k$ (implying that $\BRev(k)$ must be small). To see this, observe that the number of items for which buyer $i$ has non-zero value is a sum of $s_k$ independent random variables with expectation $s_k / \sqrt{n}$. The multiplicative Chernoff bound for large deviations (Appendix~\ref{sec:concentration} for statement) implies that with probability at most $e^{-\Omega(\ln^2 n)}$, the number of items for which buyer $i$ has non-zero value is at most $s_k\cdot \ln^2(n)/n^{1/4}$.\footnote{To see this, plug in $\varepsilon = \ln^2(n)\cdot n^{1/4}$. The resulting bound is $e^{-\Omega(\ln^2(n) n^{1/4}s_k/\sqrt{n})}$. Because $s_k \geq n^{1/4}$, this is $e^{-\Omega(\ln^2(n))}$.} So even if some buyer had the maximum possible value for every item she valued above $0$, her total value would be at most $2s_k \ln^2 (n) /n^{1/8} < s_k$. So the revenue contribution from cases where no buyer values more than $s_k\cdot \ln^2(n)/n^{1/4}$ items above $0$ is $ \leq s_k$. In the unlikely event that some buyer has non-zero value for many items, the maximum possible revenue is still at most $s_k\cdot n^{1/8}$. But this event occurs with probability $e^{-\Omega(\ln^2n)}$, and $s_k \cdot n^{1/8} \cdot e^{-\Omega(\ln^2n)} = o(s_k)$. So our total revenue is $\leq s_k$. 

So now we have shown that for all $k$, $\BRev(k) \leq 25s_k$, and therefore $\sum_k \BRev(k) \leq 25n$, and $\PRev(D) \in O(n)$. 
\end{prevproof}

\begin{prevproof}{Proposition}{prop:lb-prev-max}
%
First recall the example: $D$ has $n$ items and $m = \sqrt{n}$ buyers.  The items are partitioned into $\sqrt{n}$ disjoint sets of size $\sqrt{n}$, $S_1,\ldots, S_m$.  Each buyer $j$ has value $0$ for every item not in $S_j$, and value independently drawn from $\er_{n^{1/8}}$ for each item in $S_j$.  Our goal is to show that $\max\{\SRev(D), \BRev(D)\}\leq \PRev(D)/\Omega(\log n)$.

We note first that $\SRev(D) = n$, as each item has exactly one buyer with non-zero value, and the item can be sold to that buyer at any price $p \in [1,n^{1/8}]$ for expected revenue of $1$ (and higher revenue is not achievable). We further claim that $\BRev(D) \in O(\sqrt{n} \ln (n))$. To see this, we'll upper bound the revenue of any anonymous posted-price mechanism, and then apply Theorem~\ref{thm:chms2} to obtain a bound on $\BRev(D)$.

First observe that any posted-price mechanism with price $p > n^{5/8}$ for the grand bundle achieves revenue $0$, as the maximum possible value of any buyer for the grand bundle is $n^{5/8}$. Also, any price $p < 3 \sqrt{n} \ln (n)$ achieves revenue $O(\sqrt{n} \ln(n))$ (even if it sells with probability $1$). The remaining prices to rule out are $p \in [3\sqrt{n}\ln(n),n^{5/8}]$.

To this end, first conclude from Lemma~\ref{lem:er} (taking $\ell = \sqrt{n}, k = n^{1/8}$, and taking the bound $\ln(k)/3 \geq 1$) that for any buyer $j$, their value for the grand bundle exceeds $3\sqrt{n} \ln(n)$ with probability at most $e^{-n^{3/8}}$. Taking a union bound over all $\sqrt{n}$ buyers, we see that the probability that any buyer values the grand bundle above $3\sqrt{n} \ln(n)$ is at most $\sqrt{n} \cdot e^{-n^{3/8}} \leq e^{-n^{1/4}}$. 

As the probability that any buyer values the grand bundle above $3\sqrt{n}\ln(n)$ is at most $e^{-n^{1/4}}$ by the work above, we can immediately conclude that any anonymous posted-price mechanism with price $p \in [3\sqrt{n} \ln (n), n^{5/8}]$ achieves revenue at most $n^{5/8}\cdot e^{-n^{1/4}} = o(1)$.  Therefore, the revenue of any anonymous posted-price mechanism is $O(\sqrt{n} \ln(n))$. By Theorem~\ref{thm:chms2}, this implies that $\BRev(D) \in O(\sqrt{n} \ln(n))$ as well. We conclude that $\max\{\SRev(D), \BRev(D)\}\in O(n)$.

Finally, consider a partition mechanism that bundles each of the sets of size $\sqrt{n}$ separately and sells it to the interested buyer at price $\sqrt{n} \ln(n)/4$ (which sells with high probability by Lemma~\ref{lem:er}).  This mechanism gets a total revenue of  $\sqrt{n}\cdot \Omega(\sqrt{n}\log \sqrt{n}) = \Omega(n\log n)$.  We conclude that $\PRev(D) = \Omega(n\log n)$, and hence $\max\{\SRev(D), \BRev(D)\}\leq \PRev(D)/\Omega(\log n)$.
\end{prevproof}

\section{Omitted Proofs from Section~\ref{sec:one-cor}}
\label{app:one-cor}

\subsection{Proof of Theorem~\ref{thm:one-cor-srev-brev}}
Theorem~\ref{thm:one-cor-srev-brev} follows directly from Lemmas~\ref{lem:pointmass},~\ref{lem:symmetric}, and~\ref{lem:maincorrelated} below. We first present two helpful definitions.

\begin{definition} We say that an $n$-dimensional distribution $D$ is a \emph{point-mass in sum} distribution if there exists a $p$ such that when $\vec{v}$ is sampled from $D$, $\sum_i v_i = p$ with probability $1$.
\end{definition}

\begin{definition} We say that an $n$-dimensional distribution $D$ is \emph{symmetric} if all marginals $D_i$ are the same.
\end{definition}

\begin{lemma}\label{lem:pointmass} For any $n$-dimensional distribution $D$, there exists a point-mass in sum $n$-dimensional distribution $D'$ such that $\BRev(D')/\SRev(D') \geq \BRev(D)/\SRev(D)$.
\end{lemma}
\begin{proof}
Pick any instance $D$ whose optimal grand bundle price is $p$, and the grand bundle sells at price $p$ with probability $q$. We will transform $D$ into a point-mass in sum distribution $D'$ without decreasing the ratio $\BRev(D)/\SRev(D)$.

If $\vec{v}$ denotes a sample from $D$, then observe that we may modify $D$ to $D''$ such that $\BRev(D'') \geq \BRev(D)$ and $\SRev(D'') \leq \SRev(D)$ (and therefore $\BRev(D'')/\SRev(D'') \geq \BRev(D) / \SRev(D)$). Whenever $\sum_i v_i > p$, lower some values so that $\sum_i v_i = p$. Whenever $\sum_i v_i < p$, set all $v_i = 0$. It is clear that  $\BRev(D'') \geq \BRev(D)$, as the buyer is still willing to pay $p$ with probability $q$. It is also clear that $\SRev(D'') \leq \SRev(D)$ as we have only lowered the buyer's value for each item in a stochastically dominating way.

Next, define $D'$ to be the distribution that is exactly $D''$ conditioned on $\sum_i v_i = p$. Then $D''$ samples from $D'$ with probability $q$, and sets all values to $0$ otherwise. It is also clear that $\SRev(D'') = q\SRev(D')$, because whatever price is set for each item sells with probability exactly $q$ times the probability it sells when the buyer is drawn from $D'$ (because the buyer will never pay anything for the item if instead all values are $0$). Therefore, because $\BRev(D'') = qp$, and $\BRev(D') = p$, the two ratios are equal. That is: $\BRev(D'')/\SRev(D'') = \BRev(D')/\SRev(D')$. It is clear that $D'$ is a point-mass in sum distribution.
\end{proof}

\begin{lemma}\label{lem:symmetric} For any $n$-dimensional distribution $D$, there exists a symmetric $n$-dimensional distribution $D'$ such that $\BRev(D')/\SRev(D') \geq \BRev(D)/\SRev(D)$. If $D$ was point-mass in sum, then $D'$ is point-mass in sum as well.
\end{lemma}

\begin{proof} Define $D'$ in the following way: sample $\vec{v}$ from $D$, then randomly permute the components of $\vec{v}$ to form $\vec{v}'$. It's clear that $D'$ is symmetric. It's also clear that $\BRev(D) = \BRev(D')$. We just have to show that $\SRev(D') \leq \SRev(D)$. Let $D_i$ denote the $i^{th}$ marginal of $D$, $D'_j$ denote the $j^{th}$ marginal of $D'$, $v_i$ denote a sample from $D_i$, and $v'_j$ a sample from $D'_j$. Then $D'_j$ samples from each $D_i$ with probability $1/n$.

Now observe that $\SRev(D') = \sum_j \max_p \{p  Pr[v'_j \geq p]\}$. As each $D'_j$ samples each $D_i$ with probability $1/n$, we get that $Pr[v'_j \geq p] = \sum_i Pr[v_i > p]/n$. this means that we can rewrite $\SRev(D') = \sum_j \max_p \{p \sum_i Pr[v_i > p]/n\} = \max_p \{p \sum_i Pr[v_i > p]\}$. And observe also that $\SRev(D) = \sum_i \max_p \{ pPr[v_i \geq p]\}$. In other words, $\SRev(D')$ is exactly $\SRev(D)$ after swapping the order of the max and sum, which can only decrease $\SRev(D')$.
\end{proof}

\begin{lemma}\label{lem:maincorrelated} Let $D$ be any symmetric point-mass in sum distribution. Then $\BRev(D) \leq 5\ln (n) \SRev(D)$.
\end{lemma}
\begin{proof}
Without loss of generality, scale $D$ down so that $\SRev(D) = n$. We are essentially asking how large $\Welfare(D)$ can possibly be subject to $\SRev(D) = n$ (as $\Welfare(D) = \BRev(D)$ for point-mass in sum distributions), plus the symmetric point-mass in sum constraint. Denote $p=\Welfare(D)$.

Note that each $D_i$ is supported on $[0,p]$, and has expected revenue $1$ (because $\SRev(D) = n$ and $D$ is symmetric). So:

$$\Welfare(D_i) = \int_0^p Pr[v_i > x]dx \leq \int_0^1 dx + \int_1^p (1/x) dx = 1+\ln p$$

We now observe that we have two estimates of $\Welfare(D)$. First, we know that $\Welfare(D) = p$. And second, we know that $\Welfare(D) = \sum_i \Welfare(D_i) \leq n + n\ln p$. Putting these together, we get that $p$ must satisfy:

$$p \leq n + n \ln p$$

For all $n \geq 2$, this implies that $p \leq 5n\ln n$ (as all $p > 5n \ln n$ violate the above inequality).
\end{proof}

\subsection{Proof of Proposition~\ref{prop:lb-cor}}
\begin{prevproof}{Proposition}{prop:lb-cor}
Partition $n$ into $\ln n$ sets $S_1, \dotsc, S_{\ln n}$, where $|S_k| = n/(\ln n)$ for each $k$. Construct $D$ as follows: for each set $S_k$ independently, with probability $1-n^{-2k}$ the value of every item is 0, and with probability $n^{-2k}$ the value of each item is $n^{2k}$ times an independent draw from $\er_{n^{1/8}}$. 
It is clear that $\SRev(D)=n$, since the optimal expected revenue of each item is $1$.
To prove the claim we show that $\BRev(D) \in O(n)$ and that $\PRev(D) \in \Omega(n\ln n)$.

To bound $\BRev(D)$ from above, consider the revenue generated by a grand-bundle price of $p$.  If $p \leq n$ then this revenue is at most $n$, so assume $p > n$.  Choose $\ell \geq 0$ such that $p \in (n^{2\ell}, n^{2\ell + 2}]$.
First, observe that clearly the total value for $S_1 \sqcup \ldots \sqcup S_{\ell-1}$ is at most $n^{2\ell}$, as there are only $n$ items, and the value for each item is no more than $n^{2\ell-15/8}$. Next, observe that with probability at least $1-2n^{-2\ell-2}$, no $S_k$ for $k > \ell$ contributes any value. This is because by the union bound, the probability that any $S_k$, $k > \ell$, contributes non-zero value is at most $\sum_{k > \ell} n^{-2\ell} \leq 2n^{-2\ell-2}$. Together, this means that for any price $p \in (n^{2\ell}, n^{2\ell+2}]$, the probability of sale is completely determined, up to $\pm 2n^{-2\ell-2}$, by the probability that the value for $S_\ell$ exceeds $p - n^{2\ell}$. So we now determine the probability that the value for $S_\ell$ exceeds $p - n^{2\ell}$. 

Here, we make use of Lemma~\ref{lem:er} to claim that the buyer's value for $S_\ell$ concentrates around $n^{2\ell+1}$ whenever it is non-zero. To see this, observe that the buyer's value, when non-zero, is $n^{2\ell}$ times sum of $n/\ln n$ i.i.d. draws from $\er_{n^{1/8}}$. By Lemma~\ref{lem:er}, the probability that the sum (without the $n^{2\ell}$ multiplier) exceeds $3n$ is therefore $e^{-\Omega(\sqrt{n})}$, and the probability that the value for the bundle exceeds $3n^{2\ell+1}$ is also at most $e^{-\Omega(\sqrt{n})}$. Therefore, the probability that the bundle sells at any price $p> 3n^{2\ell+1}+n^{2\ell}$ is at most $e^{-\Omega(\sqrt{n})}+2n^{-2\ell-2}$, and therefore any such $p \leq n^{2\ell+2}$ generates expected revenue $O(1)$. Moreover, for any $p \in (2n^{2\ell}, 3n^{2\ell+1}+n^{2\ell})$, the probability of sale is at most $2n^{-2\ell-2}+n^{-2\ell}$, as the buyer must have non-zero value for \emph{some} bundle $S_k$, $k \geq \ell$. So such a price would generate revenue $O(n)$. As this argument holds for any $\ell$, we have shown that $\BRev(D) \in O(n)$.

Now we wish to show that $\PRev(D) \in \Omega(n \ln n)$. Consider the mechanism that partitions the items into $S_k$, and sets price $n^{2k+1} \cdot \ln n/2$ on $S_k$. By Lemma~\ref{lem:er}, conditioned on having non-zero value for bundle $S_k$, the buyer's value exceeds $n^{2k} \cdot (n/\ln n) \cdot \ln(n) /2$ with probability $1-e^{-\Omega(\sqrt{n})}$.  Thus the total revenue of this mechanism, and hence $\PRev(D)$, is at least $\ln(n) \cdot (n^{-2k}) \cdot (n^{2k+1} \cdot \ln (n) /2) = \Omega(n \ln n)$.
%
%
%
\end{prevproof}

\section{Computational Considerations}
\label{app:polytime}

Our main result, Theorem \ref{thm:main}, shows that $6 \cdot \max\{\SRev(D), \BRev(D)\} \geq \Rev(D)$ for a single buyer.  This suggests a simple mechanism that obtains a constant approximation to the optimal revenue: estimate $\SRev(D)$ and $\BRev(D)$, then run whichever of the two mechanisms obtains higher revenue estimate. 
In this section we argue that a slight modification of this approach can be implemented in polynomial time, given appropriate access to the distribution $D$.

We will assume that we are given a sample access to the distributions $\{D_i\}_i$. When each $D_i$ is regular,\footnote{A one-dimensional distribution is regular if $x - \frac{1-F(x)}{f(x)}$ is monotone non-decreasing.} one sample in fact suffices to find a price $p$ that guarantees revenue $\Rev(D_i)/2$~\cite{DhangwatnotaiRY15}, and $\poly(1/\varepsilon)$ suffice to find a price that guarantees $(1-\varepsilon)\cdot \Rev(D_i)$~\cite{ColeR14}. Unfortunately, sample access alone can be insufficient even when there is just $n=1$ item for arbitrary distributions (consider for instance a distribution that is $x$ with probability $1/2^{2^n}$ and $0$ otherwise, for some unknown $x$ --- there is no hope of learning a good price from few samples). In such cases we thus make the additional weak assumption that for each $D_i$ we also have access to a price $p_i$ guaranteeing revenue at least $\alpha\cdot \Rev(D_i)$, and the probability $q_i$ that $v_i \geq p_i$. Calculating $\SRev(D)$ (up to a factor of $\alpha$) and implementing a mechanism that sells items separately (and guarantees a $\alpha$-approximation to $\SRev(D)$) is then trivial: simply set the price $p_i$ for item $i$.

Given that we can compute and implement $\SRev(D)$, what we would like is to also estimate $\BRev(D)$ and compute an approximately optimal price for the grand bundle.  One approach would be to take samples from the distribution $D$, then optimize revenue for the observed empirical distribution.  
However, this strategy again suffers from the sample-complexity issues described above: the number of samples required to estimate $\BRev(D)$ might, in principle, be quite large (again, certain items may have exponentially large value with exponentially small probability).

Instead, recall that the \emph{only} reason for $\SRev(D)$ to not itself yield a 6-approximation to $\Rev(D)$ is if $\Welfare(D^C_\emptyset)$ concentrates around its expectation, in which case there exists a price $p^* = 2\Welfare(D^C_\emptyset)/5$ for the grand bundle that sells with probability at least $47/72$. Importantly, we conclude that if $\SRev(D)$ is not a 6-approximation, then some price $p^*$ for the grand bundle sells with probability at least $47/72$ and guarantees a 6-approximation.

So from here, we don't actually need to evaluate $\BRev(D)$, we just need to estimate the revenue guaranteed by any bundle price that sells with probability at least $47/72$. To test this, we simply take $\ln(1/\delta)/\varepsilon^2$ samples from each distribution, and let $G(p)$ denote the fraction of samples that exceed price $p$. The Dvoretsky-Kiefer-Wolfowitz inequality then immediately guarantees that except with probability $2/\delta^2$, $G(p)$ is within $\pm \varepsilon$ of the probability that the buyer's value for the grand bundle exceeds $p$. In particular, for any price $p$ that (really) sells with probability $q \geq 47/72$, we will have $G(p) \geq q-\varepsilon \geq q(1-2\varepsilon)$. Therefore, if we simply enumerate over all prices $p$ with $G(p) \geq 47/72 - \varepsilon$, we will find (with probability at least $1-1/\delta^2$) the best price for the grand bundle among all prices that sell with probability at least $47/72$. Moreover, we know (up to a factor of $(1-2\varepsilon)$, with probability at least $1-1/\delta^2$) the revenue generated by this price, so we can simply compare it to $\sum_i p_i q_i$ and pick whichever is better.

To summarize, we have shown that the following algorithm, with probability at least $1-\delta$, obtains a $6\alpha$-approximation in time $\poly(\ln(1/\delta),1/\varepsilon,n)$, for any $\alpha \geq (1+2\varepsilon)$.
\begin{enumerate}
\item Given as input $p_i, q_i$ such that item $i$ sells with probability $q_i$ at price $p_i$ and that $p_i q_i \geq \alpha\Rev(D_i)$.
\item Take $\ln(1/\delta)/\varepsilon^2$ samples from each $D_i$, and index them by $v_{i\ell}$. For all $\ell$, observe that $\sum_{i} v_{i\ell}$ form indepndent samples from the buyer's value for the grand bundle (so now we have $\ln(1/\delta)/\varepsilon^2$ samples from the buyer's value for the grand bundle). 
\item Let $p^*= \arg\max_{p, G(p) \geq 47/72-\varepsilon} \{p \cdot G(p)\}$. Note that $p^*$ requires enumerating over only the $\ln(1/\delta)/\varepsilon^2$ samples from Step 2. 
\item If $p^* \cdot G(p^*) \geq \sum_i p_i q_i$, bundle the items together at price $p^*$. Otherwise, sell the items separately at prices $\vec{p}$.
\end{enumerate}

\section{List of Concentration Inequalities}\label{sec:concentration}
This section contains a list of concentration inequalities used throughout the paper.
\begin{theorem}[Chebyshev's Inequality] For any random variable $X$, $\Pr[|X - \mathbb{E}[X]| \geq k] \leq \frac{\Var(X)}{k^2}$.
\end{theorem}

\begin{theorem}[Multiplicative Chernoff Bound] If $X_1,\ldots, X_n$ are independent random variables supported on $[0,1]$, and $X:= \sum_i X_i$, then for any $\varepsilon \in (0,1]$, $\Pr[|X - \mathbb{E}[X]| \geq \varepsilon \cdot \mathbb{E}[X]] \leq e^{-\varepsilon^2 \mathbb{E}[X]/3}$.
\end{theorem}

\begin{theorem}[Multiplicative Chernoff Bound for large deviations] If $X_1,\ldots, X_n$ are independent random variables supported on $[0,1]$, and $X:=\sum_i X_i$, then for any $\varepsilon \geq 1$, $\Pr[X - \mathbb{E}[X]| \geq \varepsilon\cdot \mathbb{E}[X]] \leq e^{-\varepsilon \mathbb{E}[X]/3}$. 
\end{theorem}

\end{document}